\documentclass[journal,12pt,draftclsnofoot,onecolumn]{IEEEtran}

\usepackage{adjustbox}

\usepackage{graphicx} 
\usepackage{lipsum} 
\usepackage{verbatim}
\usepackage{float}
\usepackage{colortbl}

\usepackage{cite}

\usepackage[utf8]{inputenc}
\usepackage[T1]{fontenc}

\usepackage{color}
\usepackage{multicol}
\usepackage{amsmath,mathtools,amssymb,amsfonts,amsthm} 
\usepackage{breqn}
\usepackage{cancel}
\usepackage{bbm}
\usepackage{tikz}
\usetikzlibrary{shapes.geometric}
\usepackage{enumitem}
\usepackage{algorithm2e}
\usepackage{soul}
\usepackage{caption}
\usepackage{subcaption}
\usepackage{mathrsfs}
\usepackage{xparse}
\NewDocumentCommand{\ceil}{s O{} m}{%
  \IfBooleanTF{#1} 
    {\left\lceil#3\right\rceil} 
    {#2\lceil#3#2\rceil} 
}
\DeclarePairedDelimiter{\floor}{\lfloor}{\rfloor}
\let\P\relax
\DeclareMathOperator{\P}{\mathcal{P}} 

\newtheorem{lemma}{Lemma}
\newtheorem{theorem}{Theorem}
\newtheorem{corollary}{Corollary}
\newtheorem{remark}{Remark}
\newtheorem{definition}{Definition}

\numberwithin{equation}{section} 

\usepackage{tikz}
\usetikzlibrary{automata, positioning}
\usepackage{verbatim}

\def \points {
 (-3.9523793331833419, 2.3307708246287224), (1.2067318238898894, 4.4015970064689514), (1.0374360355736894, 3.4250241358845273), (0.57263583621045067, 2.7099835335578129), (-1.1597884758414481, -0.59590668744235709), (1.4663359832054907, 0.61749599603293459), (-0.6324377801295521, 3.6116634999393047), (2.4507003358353949, 1.5654823178497179), (0.28381283945319818, -0.23736252077652822), (-1.7018789978356823, -1.2156416347399366), (-0.31936547764543505, -0.76508038143489643), (1.3065417006665208, 0.62807828788624531), (-0.22866715586648903, 2.0417462960618424), (1.874704273770166, -0.22062260288580265), (1.2496769348103447, -2.6962391719292578), (-1.273871175704377, -2.1991223246671487), (-0.86993011999477698, 0.13814756338740797), (-3.3001509609035131, -0.96880108167543311), (2.381154985856583, -0.82212151869102612), (-1.6287123320832524, -0.47000183547764024), (1.2732958025739327, 1.8993384952563097), (-0.1727482268572415, -1.5273066032468257), (0.32688962029589375, 0.78242406453607849), (3.154197313656403, 2.5407802803863282), (-0.75200376300937344, 0.8114615901518093), (1.7297759429993385, -1.9294481502748424), (1.0888019838601959, -1.3626522737329376), (-0.68099641929463628, -2.067987166383296), (-2.672699913287635, 1.282897564852804), (-4.0825604089467351, -1.3827849740409213), (-3.5756415175116798, -1.4695217588712168), (2.9969415934054453, 1.1382287188917939), (-1.8076137453425913, 3.9693466631932237), (-0.31225718741340702, -1.3200928254323452), (-3.5101354583888322, 1.6343641187488473), (0.89142314661294797, 3.3562210945443778), (0.22166971731549878, 0.069020377362348825), (2.0897373741633403, -2.1636622994265271), (-0.078512321506820984, 0.3323628354684724), (-2.1588732948570004, 3.53684424506943), (-1.9925862183126155, 0.56530366805061139), (-0.87068725335013963, -2.5800795997800892), (-2.3356931118165978, -1.962990164472584), (-1.4851832357618922, 1.3188854909404364), (-1.0628106744840817, 4.0690689296973295), (2.580541598011818, -3.689172295192654), (-1.2486890836363744, -3.8623520441015735), (1.5711606677800702, 2.2234837209237521), (-0.81674856694329212, -1.1219528271715846), (0.66972769146955691, 0.62414594400071244), (-0.64853075293804174, -2.6255019914132167), (-1.189014610372255, -2.0760662565497574), (-0.03181407975075351, 1.0866926483281376), (-1.5496533191171598, -0.83815453890389902), (-0.047380250451141442, 0.089283965441050647), (-1.2387167050520413, -3.2058178034312856), (-3.7449199896934018, 1.0563913646714085), (0.81372171377685465, -0.93650679164608941), (-2.5974647841372969, -1.5574501216175982), (-2.624115287064527, -3.3509180488447838), (-0.23308803497531713, -1.5057597977514345), (-1.3591567971149903, 2.3068334107104156), (-0.31012443225824832, -0.25221660320846384), (0.01356173166749923, 0.038487292411266), (-1.1155190034992777, -0.64725923300240962), (1.0454635053903467, -2.6320473452851485), (0.40756657556144993, 4.3000475093921446), (3.4005033528573625, 1.7734961355371386), (1.855118410123366, -1.7415593493232397), (1.4084494737818685, 0.37777374087260257), (-1.6100036589468854, 0.099332770679572066), (3.884719498375913, 1.2703166740505374), (-0.7738585577893442, -3.3161485741498025), (-4.1232672089671523, 0.028526550215038873), (2.1365436858999907, -2.2671396943886482), (0.16266682545317385, 0.64393320210981786), (1.2987274358471244, -1.0698560559857342), (2.2453642539849925, 1.0200642886933702), (-0.00083726927215389694, 0.013985285035056543), (-1.5239770198750504, 3.039905094182147), (-1.856980795130535, -0.9368584994360617), (2.4999393696323939, 1.2376099488387573), (0.91246927573620329, 0.0079663054613691824), (2.1597201688141063, 0.92095729024281936), (3.3372928856394739, 3.0416360058207745), (0.98598608787552056, 1.5333507884096234), (-2.8339627624875456, -1.3524572876390368), (-0.19125468919378943, 0.21426239711140108), (-1.6838286895251438, -1.3765425686799533), (-0.17656059047371717, 0.055894420509397268), (-1.4099849771787021, 0.23829958802680104), (-0.055652713484187398, 2.8418686014983612), (0.019872090459426131, 0.27632001711257848), (3.7494108051185679, 0.9866513643437369), (-0.11519769430777883, -0.082182090054007131), (-0.19843683888740332, 0.14812593856244149), (-0.23959773603729506, 0.6015352465897118), (0.35451871640200067, -3.4322132233430951), (-3.6293145321971449, 2.2518415008601957), (-1.6330437609015498, 1.6977160869338359)
}

\def \dispoints {
(-3.996790689498694, 2.356960793975273), (1.2268216264058387, 4.4748752716676838),
(1.3450998438508963, 4.4407551621399142), (0.44449295098410929, 2.1035508114912362), 
(-0.88946115023171224, -0.457010790056952), (1.9814718946973766, 0.8344274267570474), 
(-0.80033179018798351, 4.570456106956339), (1.8118784297793906, 1.1574094157644343), (
0.76708780041271873, -0.64154213147538242), (-1.7495202634700053, 
-1.2496714959171644), (-0.38521344177724304, -0.92282750515691203), (
0.90127040753161125, 0.43325702822666751), (-0.23929505257480016, 
2.1366417289319291), (2.1352646572029212, -0.25128637786416214), (
0.90411043673067026, -1.9506625331396192), (-1.077668367099176, -1.860411484202835), 
(-0.98762437869998798, 0.15683777159046428), (-4.4521242174974764, 
-1.306977410658837), (2.0322798924038068, -0.70166832544384572), 
(-2.0657095867317317, -0.59610729175595401), (1.1972036791083047, 
1.7858340770434242), (-0.1123898285339734, -0.99366419198947897), (
0.38549906111247062, 0.92270822792549301), (3.6134746019967832, 2.9107389612818575), 
(-0.67972450800047723, 0.73346751340704186), (1.4351876893126454, 
-1.6008548642664109), (1.3421021159094642, -1.6796612487258671), (-0.67247963288880641, -2.0421241743218594), (-1.9382741190566346, 0.93037274218198507), (-4.3947566350649794, -1.4885275000994553), (-4.2916891409578195, -1.7638039339407114), (2.0099203058091089, 0.7633612279230706), (-1.9230122581041305, 4.2227507450927355), (-0.2301896939780167, -0.97314577776729172), (-4.2063869052413665, 1.9585477281429609), (1.191101704273333, 4.4845152168408529), (0.95478776248410724, 0.29728829208462276), (1.4936312461003949, -1.5464687842541736), (-0.22989742031853061, 0.97321486637375454), (-2.4174637637814445, 3.9604884737622519), (-2.0683715332866446, 0.58680422654361819), (-0.68746056598242655, -2.0371298363676091), (-1.645915595914343, -1.3832793828890579), (-1.6076164258276016, 1.427609690149688), (-1.1725953547978065, 4.4893897284493578), (2.6595654247054004, -3.8021456773395181), (-1.427359553898593, -4.415002231471056), (1.240734566791813, 1.7558695096071155), (-0.58854060880713599, -0.80846765661028497), (0.7315635546136231, 0.68177325084010187), (-0.51557983808698471, -2.0872655390625789), (-1.0685212566720392, -1.8656801237189631), (-0.029263522513801203, 0.999571731418053), (-1.8911116086760229, -1.0228376623544837), (-0.46875506316177068, 0.88332818972350491), (-1.6723763313628206, -4.3281355577543348), (-4.4657254063086578, 1.2597208402854054), (0.65588858253607007, -0.75485771328034057), (-1.8439319056865513, -1.1056288378977652), (-2.8607909566834984, -3.6531459185417048), (-0.15297564694001875, -0.98822995878655828), (-1.0914030552727343, 1.8523874786181591), (-0.7758194612844852, -0.63095496154024433), (0.33234026177233011, 0.94315955723573042), (-0.86494428477164298, -0.5018678952083615), (0.79367388345375067, -1.9981445810359775), (0.43782568161297675, 4.6192974219593319), (4.1140907917574703, 2.1456600283307679), (1.5674987492413419, -1.4715460139346064), (0.96586050388864242, 0.25906270867876308), (-2.1459196057901506, 0.13239730165469329), (4.410193606524861, 1.4421485197327069), (-1.0544624383491281, -4.5185959064858645), (-4.6398889578292097, 0.032100763457645169), (1.4745457624316121, -1.5646772173502674), (0.24492059773815439, 0.96954314024884169), (1.6594519667984817, -1.3670110350281195), (1.9574704251951573, 0.889274723854383), (-0.059760872731196832, 0.99821272186363452), (-2.0794630273968813, 4.1479432876655142), (-1.9195459940279167, -0.96842303608050195), (1.9268138376607258, 0.95388072367521171), (0.99996189143562231, 0.0087301590187543396), (1.9776959935832421, 0.84333774785953486), (3.4293637952330212, 3.1255501851586009), (1.1628462685381979, 1.8083939160887454), (-1.9403650637709307, -0.92600400609135125), (-0.6659169138722052, 0.74602591363766857), (-1.6645581717318914, -1.3607887760121271), (-0.9533677599623932, 0.30181105722668372), (-0.9860169023266685, 0.16664533694682604), (-0.042095682057005146, 2.1495878566721003), (0.071731680774227882, 0.99742396500851349), (4.4872363649380587, 1.1808089621854492), (-0.81407417875329158, -0.58076090733377872), (-0.8013581082575908, 0.59818490647108058), (-0.3700371091416737, 0.9290169739343157), (0.47673647822282716, -4.6154438930975745), (-3.9427374866104654, 2.4463076486158872), (-1.4904798499260492, 1.5495063139479044)
}

\author{\IEEEauthorblockN{Akshit Kumar, Parikshit Hegde, Rahul Vaze, Amira Alloum, C\'edric Adjih}\\
}
\begin{document}
\title{Breaking the Unit Throughput Barrier in Distributed Systems}
\maketitle
\vspace{-.75in}
\begin{abstract}
A multi-level random power transmit strategy 
that is used in conjunction with a random access protocol (RAP) (e.g. ALOHA, IRSA) is proposed to fundamentally increase the throughput in a distributed communication network. 
A SIR model is considered, where a packet is decodable as long as its SIR is above a certain threshold.
In a slot chosen for transmission by a RAP, a packet is transmitted with power level chosen according to a distribution, such that 
multiple packets sent by different nodes can be decoded at the receiver in a single slot, by ensuring that their SIRs are above the threshold with successive interference cancelation. Since the network is distributed this is a challenging task, and we provide structural results that aid in finding 
the achievable throughputs, together with upper bounds on the maximum throughput possible. The achievable throughput and the upper bounds are shown to be close with the help of comprehensive simulations. The main takeaway is that the throughput of more than $1$ is possible in a distributed network, by using a judicious choice of power level distribution in conjuction with a RAP.

\end{abstract}

\section{Introduction}
Random access protocols (RAPs) are widely used in communication networks because of their simple distributed implementation and reasonable throughput guarantees. A particular example is the slotted ALOHA protocol that dates back to \cite{abramson1970aloha}, but is still operational for initial access acquisition in cellular, satellite, and ad hoc networks \cite{roberts1975aloha, raychaudhuri1988channel, etsi2005digital}. Some initial advances to the basic ALOHA protocol were made in \cite{choudhury1983diversity, casini2007contention,del2009high}. The basic advantage of a RAP is that it completely avoids the need of network knowledge at each node, e.g. the total number of nodes etc., and works without any coordination overhead, that can grow exponentially with the number of nodes in the network. 

The flip side of this uncoordinated communication paradigm is that packets sent simultaneously by different nodes collide, making them undecodable at the receiver node, prompting retransmissions. Thus, an important performance metric with random access is the {\it throughput} that counts the average number of packets successfully received per time slot. For the basic vanilla version of the slotted ALOHA protocol, the throughput approaches $1/e$ as the number of nodes in the network grow to infinity. Thus, the overall slot occupancy is a constant even with no coordination among nodes. 


The basic idea in most of the prior work that has addressed the question of improving the throughput of RAPs \cite{liva, stefanovic2017asymptotic, lazaro2017finite, i2018finite, sun2017coded, capture_fading, csa_paper, 4036333, 6155698, polyanskiy2017perspective, narayanan2012iterative, sun2017coded, khaleghi2017near}  is to introduce redundancy in time. In particular, each node transmits its packet in multiple slots (repetition), and packets received in collision free slots are decoded first. 
The contribution of the decoded packets is then removed from slots where collision occurs (using SIC), and all packets which are now collision free are decoded. This process is recursively continued till no slots can become collision free. This process increases the throughput fundamentally. Exploiting the random delay  with ZigZag \cite{gollakota2008zigzag}, decreases the need for repetition.

This idea of using repetition code at the transmit side and employing SIC  at
the receiver end was formalized in \cite{liva}, which called it the irregular
repetition slotted ALOHA (IRSA), where the packet repetition rate is chosen
systematically and shown to achieve a throughput close to $0.97$.
A natural extension to the IRSA was to include non-trivial forward error correction codes in contrast to just using the repetition code, which is termed as coded slotted ALOHA (CSA) \cite{csa_paper}. Using a judicious choice of codes, throughput larger than IRSA is achievable as expected, with also some analytical tractability via the iterated decoding framework typically used for LDPC codes. There is a large body of follow-up work  \cite{sun2017coded, lazaro2017finite, dovelos2017finite,  i2018finite, stefanovic2017asymptotic, narayanan2012iterative}
 to \cite{liva, csa_paper}.

Even though prior work has exploited the SIC ability intelligently, one feature that it has neglected is that if the colliding packets have different power levels, some of them can still be decodable simultaneously, without the need of decoding one of the colliding packets using some other slot. Essentially, if two packets $i,j$ with received power $P_i, P_j$ collide/overlap in the same slot, then the 
$i^{th}$ ($j^{th}$) packet is decodable as long as  the  
signal to interference ratio $SIR_i = \frac{P_i}{P_j} > \beta$ ($SIR_j = \frac{P_j}{P_i} > \beta$), where $\beta > 1$ is typically fixed threshold. 
Thus, by ensuring that the power profile of colliding packets in a slot is as different as possible, simultaneous decoding of more than $1$ packet is possible, {\bf potentially achieving a throughput of more than $1$}. Since the network in distributed, and each node has to make an autonomous decision, ensuring such received power profile is non-trivial.

In this paper, we propose a random access protocol, where nodes choose which slots to use for packet transmission similar to the slotted ALOHA, or IRSA, or CSA protocol, but employ different (random) power levels (following a distribution) for transmitting packets in slots chosen for transmission. 
For example, a node can choose between two different power levels $\{H,L\}, H>L$ to transmit. Ideally, the goal of the strategy is to ensure that in any slot, both the number of packets with power level $H$, and the number of packets with power level $L$, is roughly $1$. Since $\beta>1$, first we can decode the packet with power level $H$, and subsequently the packet with power level $L$ without needing any help from other slots. Essentially, with this protocol, a slot can potentially carry more than $1$ successfully decodable packet. We call this phenomena {\it power multiplexing}. The objective is non-trivial since the network is distributed, and each node makes it transmission decisions autonomously.
In general, the protocol can choose between $n$ different power levels for transmission in any slot with random distribution $\Delta$, to exploit as much of this power multiplexing as much as possible, together with the random slot occupancy distribution ($\Lambda$) driven by a RAP, e.g., slotted ALOHA, IRSA, CSA, and the problem is to 
 optimize the throughout over $\Delta$ and $\Lambda$. 
%
\vspace{-0.2in}
\subsection{How Multi-Level Power Transmission Leads to Fundamental Gain in Throughput}
Before presenting our main results, we illustrate how introducing multiple power levels fundamentally changes the throughput for 
any RAP.  Consider the toy examples presented in Fig. 
\ref{fig:irsa_slots_rep} and Fig. \ref{fig:irsa_pc_slots_rep}. Fig. 
\ref{fig:irsa_slots_rep} represents the case of IRSA, where one user transmits in
all the four slots, and the other two users transmit in two slots, all users with identical power $P=1$. 
Performing
SIC decoding process, in the first
iteration, only packet corresponding to user $1$ is decodable. However, in
the second iteration, the decoding process gets stuck due to the collision of
packets of user $2$ and $3$, and therefore the decoding process terminates. Hence we
are only able to successfully decode one packet over four slots resulting in a
throughput of $T = 0.25$. Meanwhile, in Fig. \ref{fig:irsa_pc_slots_rep}, employing 
two power levels $P\in \{1, 4\}$ uniformly randomly in
conjunction with IRSA, in the second iteration,  the packet
corresponding to user 2 can also be decoded because its SIR is above the threshold $\beta$. Therefore, the SIC/peeling
decoder which got stuck in the second iteration in case $1$, kickstarts
again enabled by the decoding of packet of user 2, and all the three packets belonging to the three users are decoded. 
Thus, by incurring a higher average transmit power, throughput of $T = 0.75$, thrice of what we had in the first
case with no power control, is achievable. 

What is important to note is that the increase in throughput by introducing random power levels is not additive or linear, but {\bf bootstraps the ability of the SIC process to unravel many more colliding packets}.
In Fig. \ref{fig:intro_comparison}, we present simulation results, where
SA represents slotted ALOHA, while IRSA and CSA with just two power levels are denoted as 
IRSA-DPC and CSA-DPC, respectively. It is clearly evident that there is a fundamental increase in the throughput that the proposed strategy provides over and above the slotted ALOHA (denoted as SA), or IRSA or CSA, and that too using only two power levels.

The remarkable feature is that introducing just two power levels that are used randomly in
conjunction with a RAP, {\bf throughput of more than $1$} is achievable (Figs. \ref{fig:upper_bound} and \ref{subfig:2lvlcompare}-\ref{subfig:k_approximation}), i.e., on
average more than one packet can be successfully decoded in each slot. Moreover, the proposed
strategy achieves a throughput close to double that of the relevant IRSA or CSA protocols, with just two power
levels. Of course, the power consumption with the proposed strategy is higher than CSA, however, since the system is interference limited, naively increasing the transmitted power would not yield any benefit.

To understand the benefit of increased power consumption, recall the main idea behind the improved throughput of IRSA \cite{liva} over ALOHA was to introduce repetition in transmission, thereby increasing the average power transmission and exploit the SIC capability. Thus, higher power consumption allowed higher throughput. 
CSA \cite{csa_paper} improved the power efficiency over IRSA by using more efficient codes compared to just repetition code in IRSA. 
To further exploit the SIC, as discussed before, and to decode more than one packet per slot, we need to create a power profile at the receiver where some packets are at much higher power level than the rest. The proposed strategy accomplishes this by using multiple power levels with sufficient multiplicative gaps between the levels, thus requiring larger power consumption, but allowing fundamental improvement in achievable throughputs.
\vspace{-0.2in}
\subsection{Our results} 
In this paper, we propose a multi-level random (with distribution $\Delta$) power transmit strategy 
which is used in conjunction with the RAP (with distribution $\Lambda$). With the transmit strategy, in the slot chosen according to distribution $\Lambda$, a packet is transmitted with power level $P \in {\cal P}$ chosen according to $\Delta$. The optimization problem we solve is to maximize the throughput as a function of $\Delta$ and $\Lambda$ for common RAPs, such as ALOHA, IRSA and CSA.
Our contributions are summarized as follows.
\begin{itemize}
\item For the case of ALOHA, we give an exact characterization of the throughput of the proposed strategy for two power levels, which can be directly optimized over $\Delta$.
\item With IRSA, an exact expression for throughput is not possible, however, using the expressions that are derived for iterative decoding techniques in prior work, we derive an expression for the probability of successful 
decoding of any packet asymptotically. The derived expressions significantly reduce the complexity of solving the optimization problem of maximizing the throughput over $\Delta$ and $\Lambda$. 
Similar analytical results are possible for CSA, but are omitted for brevity, 
since they do not bring any new ideas with it.
\item Since the achievable throughputs with IRSA as the RAP are not amenable to analytical expression, we derive upper bounds, that characterize fundamental limits on throughput with the proposed strategy. The derived bounds are shown to closely match the achievable throughputs via simulation. 
\item Extensive simulation results for achievable throughput are presented for the proposed strategy, with ALOHA/IRSA/CSA. We conclude that there is a fundamental improvement (close to double) in throughput with the proposed strategy at the cost of higher average power consumption. One main key observation is that {\bf throughput higher than one} is possible for a many appropriate choices of parameters, e.g. Figs. \ref{fig:upper_bound} and \ref{subfig:2lvlcompare}-\ref{subfig:k_approximation}, which is almost double  than possible for ALOHA/IRSA/CSA alone. Moreover, most of the throughput gain can be extracted via employing just two different power levels, and there is only mild incremental gain with more than two power levels. Thus, the complexity of practical implementation can be kept low by choosing only two power levels.

\item As a side result, we also analyze the throughput of the IRSA protocol under the path-loss channel model with a single power level,  which has escaped analytical tractability in the past. Towards this end, we map the IRSA protocol under the path-loss channel model without any power control to the IRSA with $n$-power levels under an ideal channel model. Leveraging results developed for analyzing the IRSA with $n$-power levels under an ideal channel model, we approximate the throughput achievable with IRSA protocol under the path-loss channel model with a single power level.
\end{itemize}

\subsection{Other Related Work}
Random access protocols have been considered in variety of different contexts and not just maximizing the throughput. For example, an important metric is the stability of the system, that defines the largest packet arrival rates at nodes, such that the sum of the expected queue lengths across different nodes remains bounded. Extensive work has been reported in studying the stability properties of slotted ALOHA protocol \cite{szpankowski1994stability, ephremides1998information, rao1988stability, luo1999stability}. Another interesting direction has been to understand the information theoretic limits of random access protocols as initiated in \cite{massey1985collision}, where the nodes may not even be slot synchronized, which was later extended in \cite{hui1984multiple, thomas2000capacity, chen2017capacity}. 

For the throughput problem, extensions of the CSA protocol have been proposed in \cite{sun2017coded} for the erasure channel,  with finite blocklength \cite{lazaro2017finite, dovelos2017finite,  i2018finite}, with multiple parallel links for each slot that do not interfere with each other \cite{stefanovic2017asymptotic}. For the basics of iterated decoding of slotted ALOHA we refer the reader to \cite{narayanan2012iterative}.

During the course of writing the paper, we found that just the raw idea of using multiple power levels chosen randomly with random access protocols can be found in a two-page short note \cite{bandai2017power}, however, without any analysis or simulation results.

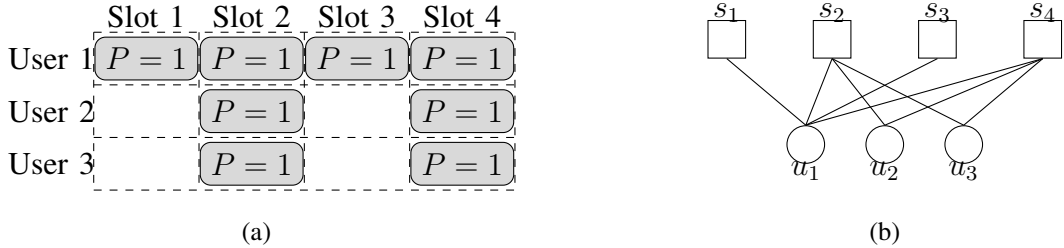
\begin{figure}[H]
	\begin{subfigure}[b]{0.5\textwidth}
	\centering
	\begin{tikzpicture}[scale = 0.7]
		\foreach \x in {0,...,-3}
		\draw[dashed] (0,\x) -- (8,\x);
		\foreach \x in {0,2,4,6,8}
		\draw[dashed] (\x,0) -- (\x,-3);
		\foreach \x in {1,2,3,4}
		\node at (2*\x-1,0.3) {Slot \x};
		\foreach \x in {1,2,3}
		\node at (-0.8, -\x + 0.5) {User \x};
		\foreach \x in {1,2,3,4}
		\node [draw, fill= gray!30, shape=rectangle, minimum width=1cm, minimum
        height=0.5cm, anchor=center, rounded corners] at (2*\x-1,-0.5) {$P =
        1$};
		\foreach \x in {2,4}
		\node [draw, fill= gray!30, shape=rectangle, minimum width=1cm, minimum
        height=0.5cm, anchor=center, rounded corners] at (2*\x-1,-1.5) {$P =
        1$};
		\foreach \x in {2,4}
		\node [draw, fill= gray!30, shape=rectangle, minimum width=1cm, minimum
        height=0.5cm, anchor=center, rounded corners] at (2*\x-1,-2.5) {$P =
        1$};
	\end{tikzpicture}
	\caption{}
	\end{subfigure}
	\begin{subfigure}[b]{0.5\textwidth}
	\centering
	\begin{tikzpicture}[scale = 0.7]
		\node[draw, shape = rectangle, anchor = center, minimum width = 0.5cm,
		minimum height = 0.5cm] at (0,2.5) (s1) {};
		\node[draw, shape = rectangle, anchor = center, minimum width = 0.5cm, minimum
			height = 0.5cm] at (2,2.5) (s3) {};
		\node[draw, shape = rectangle, anchor = center, minimum width = 0.5cm,
		minimum height = 0.5cm] at (4,2.5) (s4) {};
		\node[draw, shape = rectangle, anchor = center, minimum width = 0.5cm, minimum
			height = 0.5cm] at (6,2.5) (s5) {};
		\node[draw, shape = circle, anchor = center, minimum width = 0.5cm, minimum height = 0.5cm] at (1.5,0.5) (u1) {};
		\node[draw, shape = circle, anchor = center, minimum width = 0.5cm, minimum height = 0.5cm] at (3,0.5) (u2) {};
		\node[draw, shape = circle, anchor = center, minimum width = 0.5cm, minimum height = 0.5cm] at (4.5,0.5) (u3) {};
		\draw (u1.north) -- (s1.south);
		\draw (u1.north) -- (s3.south);
		\draw (u1.north) -- (s4.south);
		\draw (u1.north) -- (s5.south);
		\draw (u2.north) -- (s3.south);
		\draw (u2.north) -- (s5.south);
		\draw (u3.north) -- (s3.south);
		\draw (u3.north) -- (s5.south);
		\foreach \x in {1,2,3,4}
		\node at (2*\x- 2,3) {$s_\x$};
		\foreach \x in {1,2,3}
		\node at (1.5*\x,0) {$u_\x$};
	\end{tikzpicture}
	\caption{}
	\end{subfigure}
	\caption{(a) denotes a MAC frame with 3 users and 4 slots where each user
	employs a repetition code independently and uses the same power $P = 1$, 
	(b) denotes the bipartite graph representation of the same.}
	\label{fig:irsa_slots_rep}
\end{figure}

\begin{figure}[H]
	\begin{subfigure}[b]{0.5\textwidth}
	\centering
	\begin{tikzpicture}[scale = 0.5]
		\node[draw, shape = rectangle, anchor = center, minimum width = 0.3cm,
		minimum height = 0.3cm,  fill = gray!40] at (0,2.5) (s1) {};
		\node[draw, shape = rectangle, anchor = center, minimum width = 0.3cm,
		minimum height = 0.3cm, ] at (2,2.5) (s3) {};
		\node[draw, shape = rectangle, anchor = center, minimum width = 0.3cm,
		minimum height = 0.3cm, ] at (4,2.5) (s4) {};
		\node[draw, shape = rectangle, anchor = center, minimum width = 0.3cm,
		minimum height = 0.3cm, ] at (6,2.5) (s5) {};
		\node[draw, shape = circle, anchor = center, minimum width = 0.3cm,
		minimum height = 0.3cm, fill = gray!40] at (1.5,0.5) (u1) {};
		\node[draw, shape = circle, anchor = center, minimum width = 0.3cm,
		minimum height = 0.3cm,] at (3,0.5) (u2) {};
		\node[draw, shape = circle, anchor = center, minimum width = 0.3cm,
		minimum height = 0.3cm,] at (4.5,0.5) (u3) {};
		\draw[dotted] (u1.north) -- (s1.south);
		\draw[dashed] (u1.north) -- (s3.south);
		\draw[dotted] (u1.north) -- (s4.south);
		\draw[dashed] (u1.north) -- (s5.south);
		\draw[] (u2.north) -- (s3.south);
		\draw[] (u2.north) -- (s5.south);
		\draw[] (u3.north) -- (s3.south);
		\draw[] (u3.north) -- (s5.south);
	\end{tikzpicture}
	\caption{Iteration 1}
	\end{subfigure}
	\begin{subfigure}[b]{0.5\textwidth}
	\centering
	\begin{tikzpicture}[scale = 0.5]
		\node[draw, shape = rectangle, anchor = center, minimum width = 0.3cm,
		minimum height = 0.3cm,  fill = gray!40] at (0,2.5) (s1) {};
		\node[draw, shape = rectangle, anchor = center, minimum width = 0.3cm,
		minimum height = 0.3cm, ] at (2,2.5) (s3) {};
		\node[draw, shape = rectangle, anchor = center, minimum width = 0.3cm,
		minimum height = 0.3cm, ] at (4,2.5) (s4) {};
		\node[draw, shape = rectangle, anchor = center, minimum width = 0.3cm,
		minimum height = 0.3cm, ] at (6,2.5) (s5) {};
		\node[draw, shape = circle, anchor = center, minimum width = 0.3cm,
		minimum height = 0.3cm, fill = gray!40] at (1.5,0.5) (u1) {};
		\node[draw, shape = circle, anchor = center, minimum width = 0.3cm,
		minimum height = 0.3cm,] at (3,0.5) (u2) {};
		\node[draw, shape = circle, anchor = center, minimum width = 0.3cm,
		minimum height = 0.3cm,] at (4.5,0.5) (u3) {};
		\draw[] (u2.north) -- (s3.south);
		\draw[] (u2.north) -- (s5.south);
		\draw[] (u3.north) -- (s3.south);
		\draw[] (u3.north) -- (s5.south);
	\end{tikzpicture}
	\caption{Iteration 2}
	\end{subfigure}
	\caption{Graphical representation of the decoding process}
	\label{fig:irsa_decoding_rep}
\end{figure}
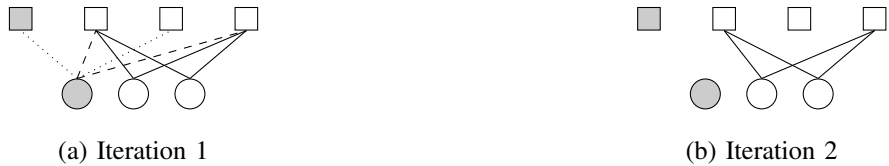
\vspace{-0.4in}
\begin{figure}[H]
	\begin{subfigure}[b]{0.5\textwidth}
	\centering
	\begin{tikzpicture}[scale = 0.7]
		\foreach \x in {0,...,-3}
		\draw[dashed] (0,\x) -- (8,\x);
		\foreach \x in {0,2,4,6,8}
		\draw[dashed] (\x,0) -- (\x,-3);
		\foreach \x in {1,2,3,4}
		\node at (2*\x-1,0.3) {Slot \x};
		\foreach \x in {1,2,3}
		\node at (-0.8, -\x + 0.5) {User \x};
		\foreach \x in {1,4}
		\node [draw, fill= gray!30, shape=rectangle, minimum width=1cm, minimum
		height=0.5cm, anchor=center, rounded corners] at (2*\x-1,-0.5) {$P =
		1$};
		\foreach \x in {2,3}
		\node [draw, fill= gray!30, shape=rectangle, minimum width=1cm, minimum
		height=0.5cm, anchor=center, rounded corners] at (2*\x-1,-0.5) {$P =
		4$};
		\foreach \x in {2}
		\node [draw, fill= gray!30, shape=rectangle, minimum width=1cm, minimum
		height=0.5cm, anchor=center, rounded corners] at (2*\x-1,-1.5) {$P =
		4$};
		\foreach \x in {4}
		\node [draw, fill= gray!30, shape=rectangle, minimum width=1cm, minimum
		height=0.5cm, anchor=center, rounded corners] at (2*\x-1,-1.5) {$P =
		1$};
		\foreach \x in {2,4}
		\node [draw, fill= gray!30, shape=rectangle, minimum width=1cm, minimum
		height=0.5cm, anchor=center, rounded corners] at (2*\x-1,-2.5) {$P =
		1$};
	\end{tikzpicture}
	\caption{}
	\end{subfigure}
	\begin{subfigure}[b]{0.5\textwidth}
	\centering
	\begin{tikzpicture}[scale = 0.7]
		\node[draw, shape = rectangle, anchor = center, minimum width = 0.5cm,
		minimum height = 0.5cm] at (0,2.5) (s1) {};
		\node[draw, shape = rectangle, anchor = center, minimum width = 0.5cm, minimum
			height = 0.5cm] at (2,2.5) (s3) {};
		\node[draw, shape = rectangle, anchor = center, minimum width = 0.5cm,
		minimum height = 0.5cm] at (4,2.5) (s4) {};
		\node[draw, shape = rectangle, anchor = center, minimum width = 0.5cm, minimum
			height = 0.5cm] at (6,2.5) (s5) {};
		\node[draw, shape = circle, anchor = center, minimum width = 0.5cm, minimum height = 0.5cm] at (1.5,0.5) (u1) {};
		\node[draw, shape = circle, anchor = center, minimum width = 0.5cm, minimum height = 0.5cm] at (3,0.5) (u2) {};
		\node[draw, shape = circle, anchor = center, minimum width = 0.5cm, minimum height = 0.5cm] at (4.5,0.5) (u3) {};
		\draw (u1.north) -- (s1.south);
		\draw (u1.north) -- (s3.south);
		\draw (u1.north) -- (s4.south);
		\draw (u1.north) -- (s5.south);
		\draw (u2.north) -- (s3.south);
		\draw (u2.north) -- (s5.south);
		\draw (u3.north) -- (s3.south);
		\draw (u3.north) -- (s5.south);
		\foreach \x in {1,2,3,4}
		\node at (2*\x- 2,3) {$s_\x$};
		\foreach \x in {1,2,3}
		\node at (1.5*\x,0) {$u_\x$};
	\end{tikzpicture}
	\caption{}
	\end{subfigure}
	\caption{(a) denotes a MAC frame with 3 users and 4 slots where each user
	employs a repetition code independently and also employ power control by
	transmitting packets at different power levels chosen independently for the
	set $\mathcal{P} = \{1,4\}$,
	(b) denotes the bipartite graph representation of the same.}
	\label{fig:irsa_pc_slots_rep}
\end{figure}
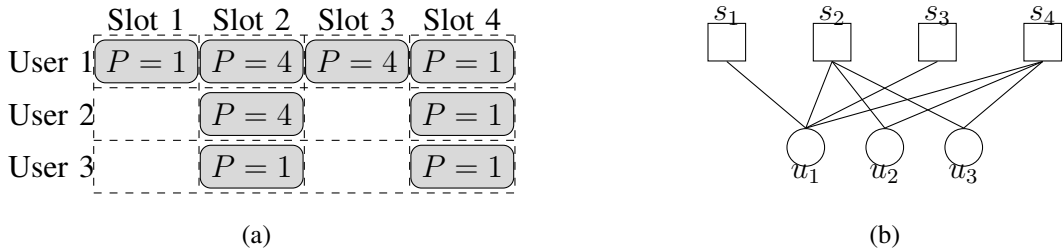

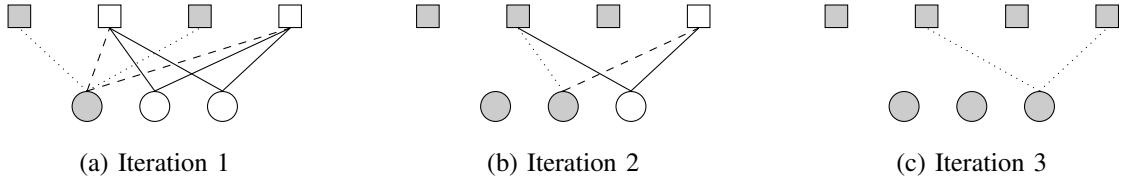
\begin{figure}[H]
	\begin{subfigure}[b]{0.32\textwidth}
	\centering
	\begin{tikzpicture}[scale = 0.6]
		\node[draw, shape = rectangle, anchor = center, minimum width = 0.3cm,
		minimum height = 0.3cm,  fill = gray!40] at (0,2.5) (s1) {};
		\node[draw, shape = rectangle, anchor = center, minimum width = 0.3cm,
		minimum height = 0.3cm, ] at (2,2.5) (s3) {};
		\node[draw, shape = rectangle, anchor = center, minimum width = 0.3cm,
		minimum height = 0.3cm, fill = gray!40 ] at (4,2.5) (s4) {};
		\node[draw, shape = rectangle, anchor = center, minimum width = 0.3cm,
		minimum height = 0.3cm, ] at (6,2.5) (s5) {};
		\node[draw, shape = circle, anchor = center, minimum width = 0.3cm,
		minimum height = 0.3cm, fill = gray!40] at (1.5,0.5) (u1) {};
		\node[draw, shape = circle, anchor = center, minimum width = 0.3cm,
		minimum height = 0.3cm,] at (3,0.5) (u2) {};
		\node[draw, shape = circle, anchor = center, minimum width = 0.3cm,
		minimum height = 0.3cm,] at (4.5,0.5) (u3) {};
		\draw[dotted] (u1.north) -- (s1.south);
		\draw[dashed] (u1.north) -- (s3.south);
		\draw[dotted] (u1.north) -- (s4.south);
		\draw[dashed] (u1.north) -- (s5.south);
		\draw[] (u2.north) -- (s3.south);
		\draw[] (u2.north) -- (s5.south);
		\draw[] (u3.north) -- (s3.south);
		\draw[] (u3.north) -- (s5.south);
	\end{tikzpicture}
	\caption{Iteration 1}
	\end{subfigure}
	\begin{subfigure}[b]{0.32\textwidth}
	\centering
	\begin{tikzpicture}[scale = 0.6]
		\node[draw, shape = rectangle, anchor = center, minimum width = 0.3cm,
		minimum height = 0.3cm,  fill = gray!40] at (0,2.5) (s1) {};
		\node[draw, shape = rectangle, anchor = center, minimum width = 0.3cm,
		minimum height = 0.3cm, fill = gray!40] at (2,2.5) (s3) {};
		\node[draw, shape = rectangle, anchor = center, minimum width = 0.3cm,
		minimum height = 0.3cm, fill = gray!40] at (4,2.5) (s4) {};
		\node[draw, shape = rectangle, anchor = center, minimum width = 0.3cm,
		minimum height = 0.3cm, ] at (6,2.5) (s5) {};
		\node[draw, shape = circle, anchor = center, minimum width = 0.3cm,
		minimum height = 0.3cm, fill = gray!40] at (1.5,0.5) (u1) {};
		\node[draw, shape = circle, anchor = center, minimum width = 0.3cm,
		minimum height = 0.3cm, fill = gray!40] at (3,0.5) (u2) {};
		\node[draw, shape = circle, anchor = center, minimum width = 0.3cm,
		minimum height = 0.3cm,] at (4.5,0.5) (u3) {};
		\draw[dotted] (u2.north) -- (s3.south);
		\draw[dashed] (u2.north) -- (s5.south);
		\draw[] (u3.north) -- (s3.south);
		\draw[] (u3.north) -- (s5.south);
	\end{tikzpicture}
	\caption{Iteration 2}
	\end{subfigure}
	\begin{subfigure}[b]{0.32\textwidth}
	\centering
	\begin{tikzpicture}[scale = 0.6]
		\node[draw, shape = rectangle, anchor = center, minimum width = 0.3cm,
		minimum height = 0.3cm,  fill = gray!40] at (0,2.5) (s1) {};
		\node[draw, shape = rectangle, anchor = center, minimum width = 0.3cm,
		minimum height = 0.3cm, fill = gray!40] at (2,2.5) (s3) {};
		\node[draw, shape = rectangle, anchor = center, minimum width = 0.3cm,
		minimum height = 0.3cm, fill = gray!40] at (4,2.5) (s4) {};
		\node[draw, shape = rectangle, anchor = center, minimum width = 0.3cm,
		minimum height = 0.3cm, fill = gray!40] at (6,2.5) (s5) {};
		\node[draw, shape = circle, anchor = center, minimum width = 0.3cm,
		minimum height = 0.3cm, fill = gray!40] at (1.5,0.5) (u1) {};
		\node[draw, shape = circle, anchor = center, minimum width = 0.3cm,
		minimum height = 0.3cm, fill = gray!40] at (3,0.5) (u2) {};
		\node[draw, shape = circle, anchor = center, minimum width = 0.3cm,
		minimum height = 0.3cm,fill = gray!40] at (4.5,0.5) (u3) {};
		\draw[dotted] (u3.north) -- (s3.south);
		\draw[dotted] (u3.north) -- (s5.south);
	\end{tikzpicture}
	\caption{Iteration 3}
	\end{subfigure}
	\caption{Graphical representation of the decoding process}
	\label{fig:irsa_decoding_rep}
\end{figure}

\begin{figure}[H]
\minipage{0.55\textwidth}
  \includegraphics[width=\linewidth]{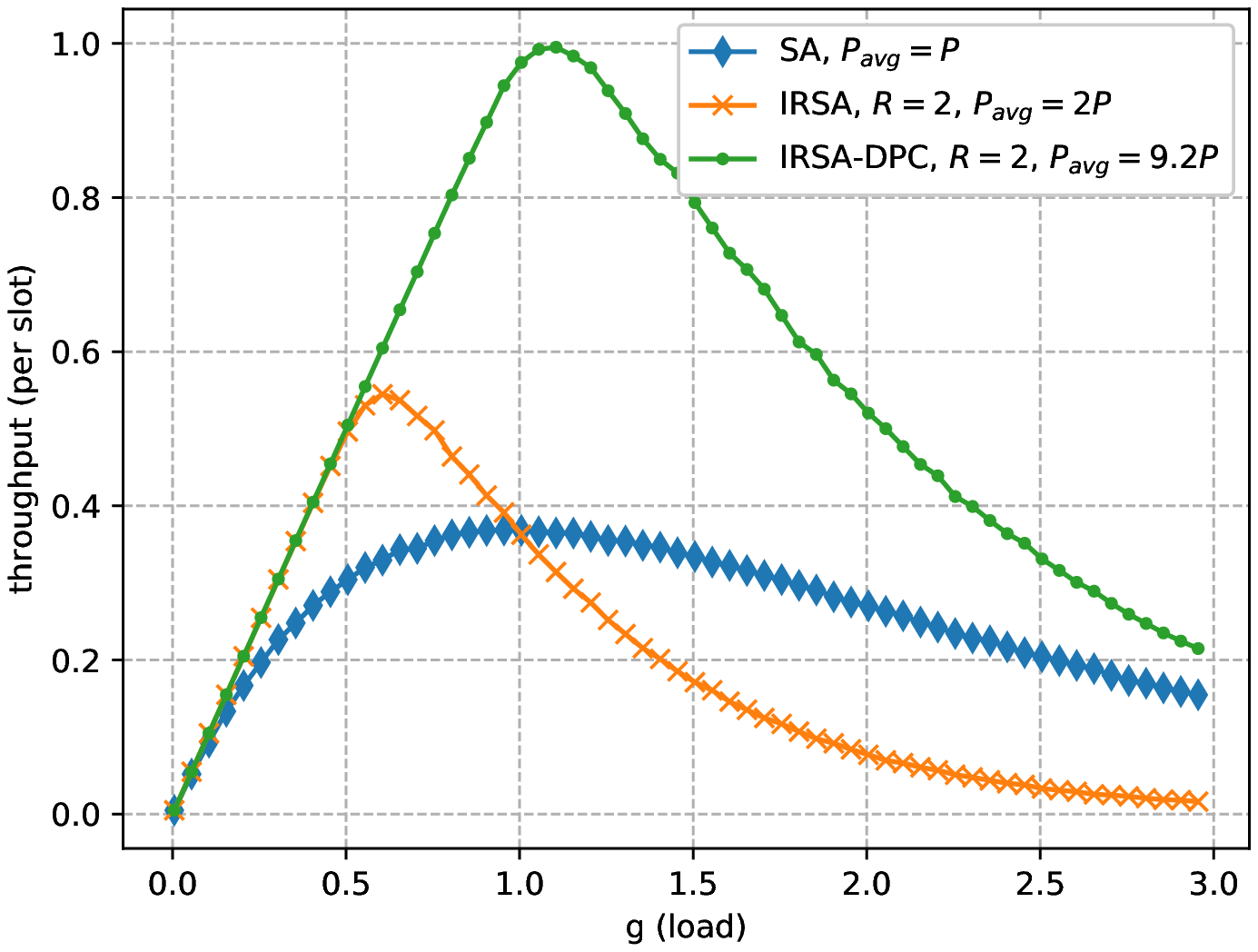}
  \label{fig:intro_irsa_irsa_pc}
\endminipage\hfill
\minipage{0.55\textwidth}
  \includegraphics[width=\linewidth]{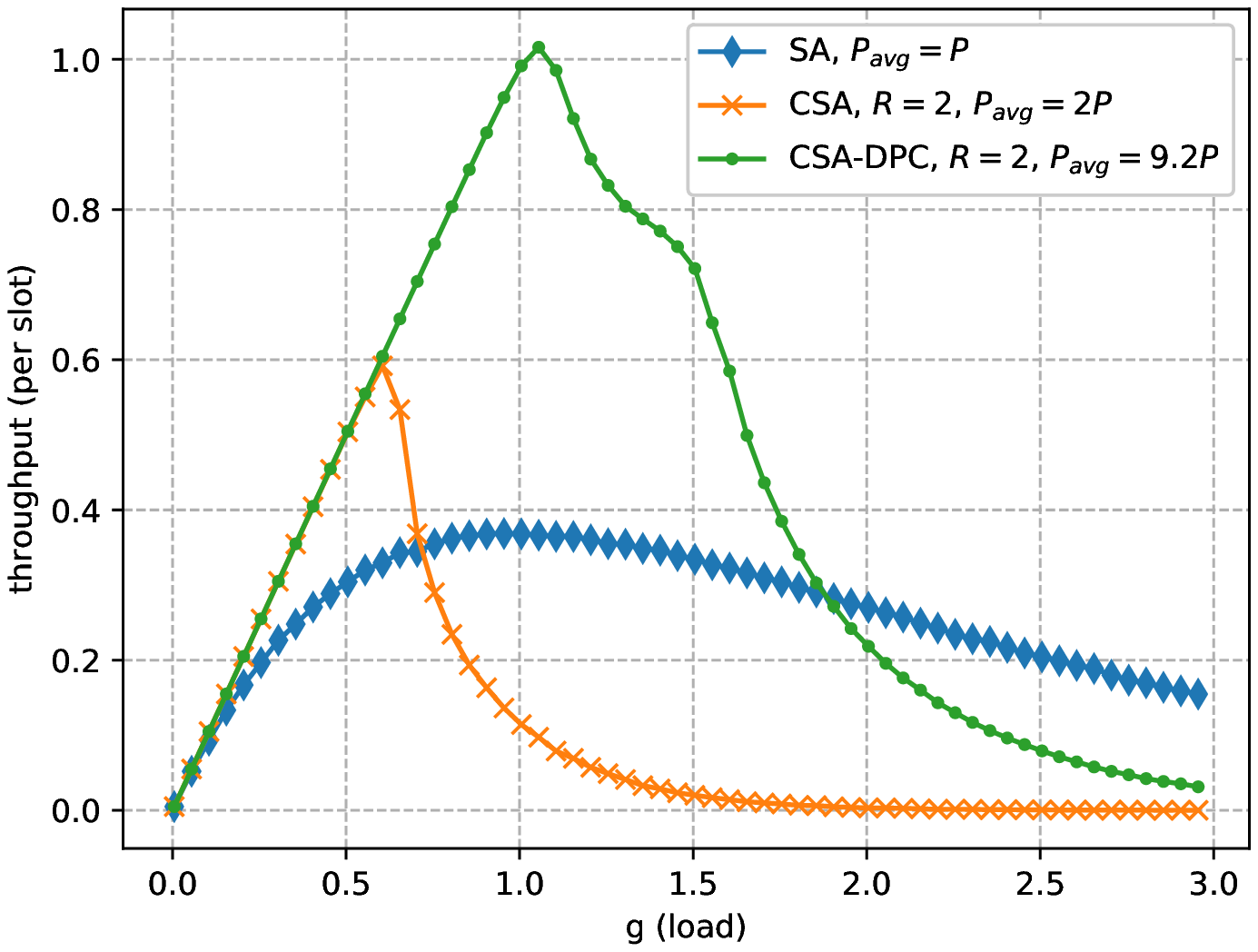}
  \label{fig:intro_csa_csa_pc}
\endminipage
\caption{$P$ is the minimum power required to ensure successful decoding at the
base station and $P_{avg}$ refers to ther average energy consumption per frame.}
\label{fig:intro_comparison}
\end{figure}

\section{System Model}
\label{sec:System-Model}
We consider a time slotted random access model, where there are total of $N$ users, that wish to communicate one packet of information in 
each frame. A frame consists of a contiguous block of $M$ time slots. All users are assumed to be slot-synchronized, and that each packet fits 
exactly in the width of one time slot. The $N$ users are uncoordinated, and have to make their transmission choices autonomously. 
We are concerned with studying large systems, and hence assume that $N$ and $M$
are large. The \emph{load} $g$ is defined as the average number of users per slot, and is
given by $g=N/M$. 


We primarily\footnote{Analysis of CSA will follow similarly.} consider the IRSA as the RAP
\cite{liva}, for selecting the slots in which each node is going to transmit its packets. In particular, with IRSA, 
each user first samples a repetition-number $l$ from a probability
distribution $\Lambda = \{\Lambda_{l}\}_{l=1}^{\text{deg}}$ (where $l$ ranges from 1 to the maximum
allowed repetition number  called $\text{deg}$). The user then creates $l$ replicas of
its packet and transmits each of them uniformly at random in $l$ unique slots
in a frame of total $m$ slots. Define
$R$ as the average repetition rate, $R \triangleq \sum_{l} l\Lambda_{l}$. 
Each replica also stores the location of the other packet-replicas.


For each of the replicas to be transmitted for each node, the transmit power is chosen as follows. 
Let $\P = \left\{ P_k \right\}_{k=1}^{|\P|}$ denote the set of power levels that any user can
select for transmission, and $ \Delta = \{ \delta_{k} \}_{k=1}^{|\P|}$, ($\sum_{k=1}^{|\P|} \delta_{k} = 1$) denote the probability
distribution over the power levels of set $\P$, which we will call the  {\it power-choice
distribution}. 
The transmitted power $P$ for each of the replicas (that is already selected according to $\Lambda$) is sampled independently according to the distribution defined as 
$\Pr(P = P_k) = \delta_k, \forall k \in \left\{ 1, \dots, |\P| \right\}$
(without loss of generality, assume that $P_{k} > P_{k+1}$).  The overall strategy is called the IRSA-PC protocol.

Therefore, the 
average power per replica is $\sum_{k} P_{k}\delta_{k}$, and the average power
per user is $P_{avg} = R \sum_{k} P_{k}\delta_{k}$. All the users use the common
repetition distribution $\Lambda$, set with powers $\mathcal{P}$ and power-choice
distribution $\Delta$.


The communication channel is assumed to be interference limited, and hence we ignore the additive noise as done in the prior work \cite{liva}. 
We consider an {\it ideal channel} for analysis similar to prior work \cite{liva}, where the received power $P_{\text{rec}}$ is 
equal to the transmitted power $P$ i.e., $P_{\text{rec}}=P$. 

A packet can be decoded at the receiver, if its \emph{signal to interference
ratio (SIR)} is sufficiently high. Let the set of users who have transmitted
their packets in slot $m$ be denoted by $\mathcal{R}^{0}_{m}$. Then, the
$SIR^{(i,m)}$ of a user $i \in \mathcal{R}^{0}_{m}$ is given by: 
	$SIR^{(i,m)} = \frac{P_{rec}^{(i)}}{\sum_{j \in \mathcal{R}^{0}_{m}\setminus i}
                    P_{rec}^{(j)}}$. Unless otherwise specified, we consider the ideal channel model for the rest of the paper.

The packet of user $i$ (called packet $i$ here after) can be decoded  in slot $m$ if $SIR^{(i,m)} \geq \beta$, where
$\beta$ is the \emph{capture threshold}. Throughout this paper, we assume that $\beta > 1$, 
which implies that
at most one packet can be captured in a slot at a particular iteration. The decoding
process follows an iterative procedure: at a particular iteration, as many packets with $SIR^{(i,m)} \geq \beta$ are
decoded,  after which all the decoded packets are subtracted from
all the slots that they occupy, and the process moves onto next iteration. This process
terminates when no more new packets can be decoded. Such a process is called
\emph{successive interference cancellation (SIC)} in literature. 


In this paper, a system is represented by the parameters $\left\{ M, \P, \Delta, \beta, \Lambda \right\}$, where
$M$ is the number of slots per frame, $\P$ is the set of power levels, $\Delta$ is the
associated power-choice distribution, $\beta$ is the capture-threshold and $\Lambda$ represents the
RAP. 
The \emph{throughput} $T(g, M, \P, \Delta, \beta, \Lambda)$ of the system is the average number of user-packets
decoded per slot for a specific load $g = N/M$. A throughput $T(g)=T$ is said to be achievable in a system
if a throughput $T$ can be obtained on the system for some load $g$.

\begin{definition} (Capacity)
    For a system $\left\{ M, \P, \Delta, \beta, \Lambda \right\}$, the capacity $T^\star(M, \mathcal{P},\Delta, \beta, \Lambda)$ is
defined as  the throughput $T(g)$ maximized over all possible loads $g$, 
\begin{equation}
    T^\star(M, \mathcal{P}, \Delta, \beta, \Lambda) := \sup_{g \geq 0} T(g, M, \P, \Delta, \beta, \Lambda).
\label{eq:capacity_1}
\end{equation}
\end{definition}

Our theoretical analysis focuses on large systems, with 
$N, M\rightarrow \infty$. Therefore, we define asymptotic versions of the throughput and
capacity as well.

\begin{definition}(Asymptotic Throughput)
	\label{def:asymptotic_throughput}
    The asymptotic throughput $T_\infty(g,\P, \Delta, \beta, \Lambda)$ is defined as:
        $T_\infty(g,\P, \Delta, \beta, \Lambda) = \lim_{M \rightarrow \infty} T(g, M , \P, \Delta,
        \beta,  \Lambda)$.\end{definition}
\begin{definition}(Asymptotic Capacity)
	\label{def:asymptotic_capacity}
    The asymptotic capacity $T^\star_\infty(g,\P, \Delta, \beta, \Lambda)$ is
    defined as:
        $T^\star_\infty(\P, \Delta, \beta, \Lambda) = \sup_{g\geq 0} \lim_{M
        \rightarrow \infty} T(g, M , \P,
        \Delta, \beta, \Lambda)$.
    \end{definition}
\subsection{Graphical Model Description}
\label{subsec:graphical_model}
To facilitate analysis, we next describe the graphical model description of the proposed strategy as follows. 
One frame of the IRSA-PC model can be represented by a
bipartite graph $\mathcal{G} = \{U,S,E\}$, where $U$ is the set of $N$ nodes
representing the users, $S$ is the set of $M$ nodes representing the slots and
$E$ is the set of edges. An edge $e = (u,s)$ exists between a user node $u \in U$ and a
slot node $s \in S$
if the user $u$ has transmitted a packet-replica to the slot $s$ in the given
frame. If the user $u$ transmits a packet-replica to slot $s$ with power
$P$, then we will say that the edge connecting $u$ and $s$ has power $P$. Fig. \ref{fig:csa_graph_model} shows the
construction of the bipartite graph, for the example in Fig. \ref{fig:csa_time_model}.

In a
particular iteration of the SIC algorithm, we say that the egde $e = (u,s)$ is {\bf known} if the user
$u$'s packet-replica in slot $s$ has been decoded by the SIC algorithm. Otherwise, the edge is
unknown. 
This construction is very similar to the Tanner graph representation of
LDPC codes, with the user nodes acting like bit nodes, and the slot nodes
acting like check nodes. In fact, the
SIC algorithm can be interpreted as a peeling
decoder of LDPC codes over erasure channel: in a slot, if a particular
user $u$'s packet is decoded, then the node $u$ and all the edges connected to
it are deleted (made ``known'') from $\mathcal{G}$. This process is repeated until either all
edges are deleted, or no more can be deleted. This interpretation opens up an avenue to
use the powerful density evolution (DE) techniques from LDPC codes to analyze our IRSA-PC model. In order to describe the
DE analysis in later sections, we
need a modicum of notation which we will define now.
\begin{figure}
    \centering
    \includegraphics[scale=0.5]{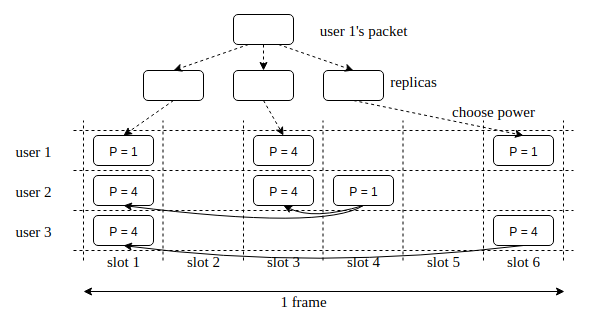}
    \caption{Illustration of the IRSA-PC model with $N=3$ and $M=6$. The set of
    powers is $\mathcal{P} = \{1,4\}$, and the powers chosen by the users are
    indicated in the Fig.. The capture threshold $\beta$ is 2. In slots 1 and 3, none of the users have $SIR \geq \beta$ in the first iteration. User 2's packet is
captured in slot 4 as it has no collisions. User 3's packet is captured in slot
6 as its $SIR \geq \beta$. In the next iteration, users 2 and 3's packets are
subtracted from slot 1 by inter-slot cancellation which enables the capture of
user 1's packet. Also, in slot 3, user 2's packet is subtracted using
intra-slot cancellation which enables the capture of user 1's packet in slot 3
as well.}
    \label{fig:csa_time_model}
\end{figure}

For our purposes, it suffices to summarize the bipartite-graph by its user (and
slot) -node degree
distribution, which denotes the fraction of user (and slot) nodes with a particular degree.
With our IRSA-PC protocol, the probability that a user creates $l$
replicas of his packet (which means that the corresponding user-node has a degree $l$ in the
bipartite graph) is equal to $\Lambda_{l}$. In the large system limit, the
user-node degree distribution of the graph will be nearly identical to this
probability distribution $\Lambda_{l}$. Therefore, we use $\{\Lambda_{l}\}$ to denote both the
probability distribution of the number of packet replicas, as well as the user-node degree
distribution of the bipartite-graph.
 Similarly, let $\{\psi_{l}\}$ represent
the slot-node degree distribution, where $\psi_{l}$ represents the
probability that a slot node has a degree $l$. These distributions can be
succinctly represented in polynomial form as: $\Lambda(x) = \sum_{l} \Lambda_{l}x^{l}$
and $\psi(x) = \sum_{l} \psi_{l}x^{l}$. It is useful to also define edge-perspective
degree distributions. Let $\lambda_{l}$ be the probability that an edge is
connected to a user node of degree $l$. Similarly, let $\rho_{l}$ be the
probability that an edge is connected to a slot node of degree $l$. Their
corresponding polynomial representations are: $\lambda(x) = \sum_{l}
\lambda_{l}x^{l-1}$ and $\rho(x) = \sum_{l} \rho_{l}x^{l-1}$. The edge and node
perspective distributions can be computed from each other with the following
simple relations: $\lambda(x) = \Lambda'(x)/\Lambda'(1)$, $\rho(x) =
\psi'(x)/\psi'(1)$, $\Lambda(x) = \int_{0}^{x} \lambda(x)/ \int_{0}^{1}
\lambda(x)$, and $\psi(x) = \int_{0}^{x} \rho(x)/ \int_{0}^{1} \rho(x)$. Since
the slots where packets are transmitted by each user are chosen uniformly at random, the slot node-perspective and edge-perspective degree
distributions are determined by the load $g$ and repetition rate $R$ as: $\psi(x)
= \rho(x) = e^{-gR(1-x)}$ (for the derivations, refer \cite{liva}).

\begin{figure}[H]
\centering
    \begin{tikzpicture}[scale = 0.6, square/.style={regular polygon,regular polygon sides=4}]
        \node at (0,0) [circle,draw, minimum size = 0.75cm] (u1) {};
        \node at (2,0) [circle, draw, minimum size = 0.75 cm] (u2) {};
        \node at (4,0) [circle, draw, minimum size = 0.75 cm] (u3) {};
        \node at (0,0) {$1$};
        \node at (2,0) {$2$};
        \node at (4,0) {$3$};
        \node at (2,-1) {$\textrm{Users}$};
        \node at (-1,2) [square,draw,minimum size = 1cm] (s1) {};
        \node at (0.2,2) [square,draw,minimum size = 1cm] (s2) {};
        \node at (1.4,2) [square,draw,minimum size = 1cm] (s3) {};
        \node at (2.6,2) [square,draw,minimum size = 1cm] (s4) {};
        \node at (3.8,2) [square,draw,minimum size = 1cm] (s5) {};
        \node at (5,2) [square,draw,minimum size = 1cm] (s6) {};
        \node at (-1,2) {$1$};
        \node at (0.2,2) {$2$};
        \node at (1.4,2) {$3$};
        \node at (2.6,2) {$4$};
        \node at (3.8,2) {$5$};
        \node at (5,2) {$6$};
        \node at (2,3) {$\textrm{Slots}$};
        \draw[] (u1.north) -- (s1.south);
        \draw[] (u1.north) -- (s3.south);
        \draw[] (u1.north) -- (s6.south);
        \draw[] (u2.north) -- (s1.south);
        \draw[] (u2.north) -- (s3.south);
        \draw[] (u2.north) -- (s5.south);
        \draw[] (u3.north) -- (s1.south);
        \draw[] (u3.north) -- (s6.south);
    \end{tikzpicture}
    \caption{Graph construction for the example in Fig.
	\ref{fig:csa_time_model}.}
    \label{fig:csa_graph_model}
\end{figure}
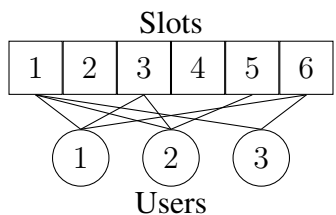

\section{Slotted Aloha}
\label{sec:SA}

Before analyzing IRSA-PC, we characterize the throughput (Theorem \ref{theorem:dpc_sa_thruput}) when PC is used with Slotted Aloha (SA) 
that is a special case of IRSA-PC with $\Lambda(x)=1$.  Recall that without PC, the throughput of SA is $T_{SA} = ge^{-g}$, with a
capacity of $1/e$ achieved at $g=1$.
Consider a system, called SA-DPC (stands for Slotted Aloha with dual power control) with parameters 
$\left\{ M, \P = \left\{ P_1, P_2 \right\},
\Delta = \left\{ \delta, 1-\delta \right\}, \beta, \Lambda = \{1\} \right\}$, where a user wishes to transmit a single packet with no replicas, and 
the transmit power for the single packet is chosen as $P_1$ with probability $\delta$, and $P_2$ with
probability $1-\delta$. Once the transmit power is chosen, the packet is transmitted in one of the $M$ slots chosen uniformly at random.

With $P_1>P_2$, define $k \triangleq P_{1}/\beta P_{2}$. With
$\beta > 1$, this implies that decoding of a packet in a slot with collisions is possible only if it has power $P_{1}$, 
and there are a maximum of $k$ other packets in the slot, all with lower transmit power $P_{2}$. The following lemma gives a bound on the probability
of the event that there are $k$ packets transmitted in a single slot by $N$
nodes. The proofs of this section are can be found in Appendix 
\ref{appendix:proof_SA}.
\begin{lemma}
    Consider a system with $M$ slots, and let there be $N$ users trying to communicate packets. For
    the fixed load $g = N/M$, the probability that a slot has $k$ packets transmitted in it (i.e.,
    each packet in that slot has $k-1$ interfering packets) under the Slotted Aloha scheme(either
    vanilla SA or SA-DPC) is $\mathcal{O}\left( \frac{1}{\textrm{poly}(k)}\right)$.
\label{lm:k_to_inf_approx}
\end{lemma}

\begin{remark}
\label{rmk:k_to_inf_approx}
Lemma \ref{lm:k_to_inf_approx} points out that as $k$ increases, the probability of having $k$ packets transmitted in some slot $j$
decreases at least as fast as $\frac{1}{k^p}$ for some finite $p$. Thus, for a slot to have at most $k$ other transmissions with power $P_2$ with small probability, $k$ has to be chosen suitably large.
However, notice that in the SA-DPC, we have defined that $P_1 = k\beta P_2$, hence taking
$k$ large enough requires $P_1$ to be very large, which is an unrealistic assumption to
make.
Hence we choose $k$ to be some integer such that $P_1$ is not unrealisably large
and at the same time the probability of $k$ or more packets getting transmitted together in one
slot is low ($\sim 10^{-2}$). After some experimentations, we found $k = 5$ to
be a suitable value. So while all our
results would follow by taking $k \to \infty$, for the purposes of constructing
the set $\mathcal{P}$ with powers (details in upcoming sections) and for
experiments, we would take $k = 5$ and get a very close approximation to
the case of $k \to \infty$.
\end{remark}
We will now describe a recursive method to compute the throughput in the SA-DPC
scheme. 

\begin{lemma}
    Consider a SA-DPC system,  
    with power levels satisfying
    $P_1 \geq k \beta P_2$ for some integer $k$. If $k$ is sufficiently large (Remark \ref{rmk:k_to_inf_approx}), then the probability
    of a packet with power $P_1$ being decoded in its slot is $g \delta e^{-g\delta}$, where
    $g=N/M$ is the load.
\label{lm:count_capture}
\end{lemma}
\begin{proof}[Proof Sketch]
The set of packets transmitted in the same slot by different nodes are called interferers for each other.
    Here we will provide a proof sketch following an intuition from Lemma \ref{lm:k_to_inf_approx}.
    A packet with transmit power $P_1$ can get decoded in a slot if it has upto $k$ interferers with transmit power $P_2$.
    Following Lemma \ref{lm:k_to_inf_approx}, as $k$ is made large, the probability that a packet
    has greater than $k$ interferers becomes very small. Therefore, for an appropriately large enough choice of $k$, with high probability, a
    packet with power $P_1$ can be decoded in a slot, as long as there is no other packet with 
    power level $P_1$ transmitted in the same slot. This is equivalent to a system with load $g\delta$, where $\delta$ is the probability of choosing power level $P_1$ for any node. Therefore,
    the probability that a packet with power $P_1$ is
    decoded in any slot is $g\delta e^{-g\delta}$.
\end{proof}

\begin{theorem}
\label{theorem:dpc_sa_thruput}
Consider an SA-DPC system, with the power levels being such that $P_1 \geq k\beta P_2$, for some
integer $k$. For sufficiently large $k$ and $N,M \rightarrow
\infty$), the throughput of the system is $T_{\text{SA-DPC}} = g\delta e^{-g\delta} + (1 + g\delta)g(1-\delta)e^
{-g}$, where $g = N/M$ is the load.
\end{theorem}

\begin{proof}[Proof Sketch]
    From Lemma \ref{lm:count_capture}, we have the throughput contribution from packets with transmit power
    $P_1$. A packet with power $P_2$ cannot be decoded in slots when the number of packets with power $P_1$ is greater than
    $1$. Probability
    of a slot having no packet with power $P_1$  is $ (1-1/M)^{\delta N} \approx e^{- g
    \delta}$. From Lemma
    \ref{lm:count_capture}, probability of a slot having 1 packet with power $P_1$ is $g\delta e^{-g\delta}$.
    Hence, the probability that a packet with power $P_2$ in a slot faces no interference from a 
    packet with power $P_1$ is $g\delta e^{-g\delta} + e^{-g\delta}$.
    Since there are about $N(1-\delta)$ users that transmit a packet with power $P_2$, the probability
    that exactly one of them transmits a packet in any slot is $g(1-\delta)e^{-g(1-\delta)}$. Hence,
    probability that a packet with power $P_2$ is decoded in a slot is $ (1 + g\delta)g(1-\delta)e^
    {-g}$.
\end{proof}

Full proof of Lemma \ref{lm:count_capture} and Theorem 
\ref{theorem:dpc_sa_thruput} can be found in Appendix \ref{appendix:proof_SA}.

\begin{remark}Setting $\delta = 0$ or $\delta = 1$, we get the well-known result that for Slotted Aloha without PC,
$T_{\text{SA}} = ge^{-g}$ which is maximised for $g = 1$ giving $T^\star_{
\text{SA}} = 1/e$.
\end{remark}
Using Theorem \ref{theorem:dpc_sa_thruput}, we can optimize $\delta$ to obtain the maximum
throughput $T_{SA-DPC}$ for the SA-DPC, for a given value of the load. The maximum value of throughput possible is $0.658$, which is
achieved for a load of $g=1.75$ and $\delta=0.4$.
Let $\Delta T := T_{\text{SA-DPC}} - T_{\text{SA}}$ be the throughput gain with SA-DPC over SA. Then for any non-trivial
$\delta$, we have that $\Delta T = g^2\delta(1-\delta)e^{-g}$. Hence $\forall
\delta \in (0,1), \Delta T > 0$, which implies that for all loads, choosing a
SA-DPC with two power levels results in throughput gain over vanilla Slotted Aloha which
can be seen in Fig. \ref{fig:compare_sa_sa_dpc}.
\begin{figure}
\centering
\includegraphics[scale = 0.5]{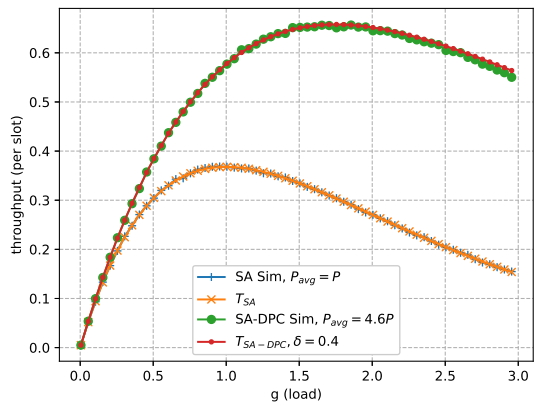}
\caption{Plot comparing vanilla SA and
SA-DPC}
\label{fig:compare_sa_sa_dpc}
\end{figure}

\subsection{$n$ Level Power Control in Slotted ALOHA}
In the $n$-level power control scheme for slotted aloha (SA-nPC), we generalise
the SA-DPC scheme from two power levels to $n$-power levels. In this
generalization, users can now choose to transmit their packet with power $P_i
\in \mathcal{P}$ where $|\mathcal{P}|=n$. Assume that $P_i \geq k \beta P_{i+1},
\forall i \in [n-1]$, where $[m] \triangleq \{1,2,\dots,m\}$

\begin{theorem}
\label{theorem:npc_sa_thruput}
Consider an SA-nPC system $\left\{ M, \P = \left\{ P_1, \dots, P_n \right\}, \Delta = \left\{
        \delta_1,\dots, \delta_n
\right\}, \beta, \Lambda = \{1\} \right\}$, with the power levels being such that $P_i \geq k\beta P_{i+1}$, for all
$i < n$ and for some
integer $k$. Let the number of users be $N$. For sufficiently large $k$ and number of users and
slots being large($N,M \rightarrow
\infty$), the throughput of the system is $T_{\text{SA-nPC}} = \sum_{i=1}^n \left(\prod_{j=1}^{i-1}(1 +
g\delta_j)e^{-g\delta_j}\right)g\delta_i e^{-g \delta_i}$, where $g = \frac{N}{M}$ is the load.
\end{theorem}

Proof is similar to that of Theorem \ref{theorem:dpc_sa_thruput} and can be
found in Appendix \ref{appendix:proof_SA}.

\section{IRSA-PC}
\label{sec:irsa-pc}

In this section, we will formulate a theoretical analysis of the IRSA-PC protocol using the graphical model representation (Section
\ref{subsec:graphical_model}), where for ease of explanation, we only  consider that each node uses
only two power levels (which we denote as IRSA-Dual-Power-Control or IRSA-DPC). 
The analysis presented in this section can be generalized to
$n$ power levels, called the IRSA-nPC protocol, with details provided in  Appendix \ref{appendix:irsa-npc}.
Consider a system $\{M,\P, \Delta, \beta, \Lambda \}$ and a given load $g$, where we are interested in finding its capacity. We
begin by discussing the problem of computing the \emph{throughtput} $T
(g,M,\P,\Delta,\beta,\Lambda)$, which in itself is a non-trivial problem. In
particular, we are interested in computing the asymptotic throughput $T_{\infty}
(g,\P,\Delta,\beta,\Lambda)$ given in Definition \ref{def:asymptotic_throughput}.
\begin{definition}\label{def:pl}
Let $P_L(g,\P,\Delta,\beta,\Lambda)$
denote the per-slot packet loss probability, i.e. the probability that a user
packet is not decodable by the receiver in some slot at the end of the decoding
process. 
\end{definition}

From Definition \ref{def:asymptotic_throughput}, it follows that
\begin{align}
\label{def:asympototic_throughput_packet_loss_prob}
T_{\infty}(g,\P,\Delta,\beta,\Lambda) &= g\left(1 - P_{L}
(g,\P,\Delta,\beta,\Lambda) \right).
\end{align}
Following the 
graphical model representation (Section
\ref{subsec:graphical_model}), during the SIC decoding process, let $p_i$ denote the
probability that an edge connected to a slot node is unknown at iteration $i$,
and $q_i$ denote the probability that an edge connected to a user node is unknown at
iteration $i$.
The DE generates a sequence of probabilities $q_1 \to p_1 \to q_2 \to p_2 \to
\dots \to p_\infty$ and the following lemma characterizes this
sequence.

\begin{lemma}[Density evolution recursion for IRSA-DPC]
    Consider an IRSA-DPC protocol \newline
    $\left\{ M, \P = \left\{ P_1, P_2 \right\},
    \Delta = \left\{ \delta, 1- \delta \right\}, \beta, \Lambda \right\}$, and $R$
    is the average repetition rate defined in Section \ref{sec:System-Model}.
    Let $P_1 \geq k \beta P_2$ for some integer $k$. If $k$ is large enough, for a 
constant load $g = N/M$ and $M \rightarrow \infty$, the sequences of probabilities $\left\{ p_i
\right\}_{i\geq 1}$ and $\left\{ q_i \right\}_{i\geq 1}$ evolve as follows,
\begin{align}
\label{eq:q_1}
\text{(Initial Condition)} \quad q_1 &= 1, \\
\label{eq:p_iteration}
\text{(IRSA-DPC Slot Node DE Update)} \quad p_i &\approx f_p(q_i) , \quad i\geq 1 \\
\label{eq:q_iteration}
\text{(IRSA-DPC User Node DE Update)} \quad q_{i+1} &= f_q(p_i), \quad i\geq 1
\end{align}
where, $f_q(p) = \lambda(p)$ and $f_p(q) = 1 - (1 - \delta)e^{-g\delta q}
- \delta e^{-g q \delta R} - \delta(1-\delta)g q Re^{-gqR}$.
\label{lm:p_irsa_2lvl}
\end{lemma}
Proof of Lemma \ref{lm:p_irsa_2lvl} can be found in Appendix \ref{app:lemDE}.

Let $p_\infty$ denote the probability that an edge connected to a slot node is
unknown at the end of the SIC decoding process. Note that in all practical
settings, $p_{{\mathsf L}} \approx p_\infty$ for some large but finite number of iterations ${\mathsf L}$.
Let at the start of the SIC decoding process, a particular user node
has $l$ edges which are all unknown. The packet corresponding to that user is
non-decodable if and only if at the end of the SIC
decoding process, the user still has $l$ edges which are all unknown. 
Hence the probability that the packet of user with degree $l$ is not decoded is
the probability that all the $l$ edges are unknown and due to the implicit
assumption of independence made in the DE analysis, the probability of packet
(of the user with $l$ edges) not being decoded is $\prod_
{i= 1}^l p_\infty = p_\infty^l$. Averaging over the distribution
of edges of the user node
$\{\Lambda_l\}$, the average packet loss probability (Definition \eqref{def:pl}) is 
\begin{align}
P_{L}(g,\P,\Delta,\beta,\Lambda) = \sum_{l}\Lambda_l p_\infty^l = \Lambda
(p_\infty).
\label{eq:packet_loss_prob}
\end{align}
Therefore, we can rewrite \eqref{def:asympototic_throughput_packet_loss_prob} as
\begin{align}
T_{\infty}(g,\P,\Delta,\beta,\Lambda) &= g\left(1 - \Lambda(p_\infty)\right).
\label{eq:asymptotic_throughput}
\end{align}

From \eqref{eq:asymptotic_throughput} and Lemma \ref{lm:p_irsa_2lvl}, we
infer that 
given the load $g$ and large system $\{M,\P,\Delta,\beta,\Lambda\}$ with $M \to
\infty$ and $\P = \{P_1,P_2\}$ such that $P_1 \geq k\beta P_2$, we can
compute $T_{\infty}(g,\P,\Delta,\beta,\Lambda)$. For example, the curve \emph{IRSA-DPC DE}
in Fig. \ref{subfig:2lvlcompare} plots $T_{\infty}(g,\P,\Delta,\beta,\Lambda)$
as a function of the load $g$ for the case of $\P = \{10,1\}, \Delta = 
\{0.4,0.6\},\beta = 2, \Lambda(x) = 0.5x^2 + 0.28x^3 + 0.22x^8$.

Given that we have an expression for the asymptotic throughput $T_\infty
(g,\P,\Delta,\beta,\Lambda)$ in \eqref{eq:asymptotic_throughput}, we can think
of optimizing $T_\infty(g,\P,\Delta,\beta,\Lambda)$ with respect to different
parameters like $g$ and $\Lambda$ to find the optimal asymptotic throughput. First let us
consider the case of
optimizing $T_\infty(g,\P,\Delta,\beta,\Lambda)$ for a fixed system $
\{M,\P,\Delta,\beta,\Lambda\}$ with the load $g$ as the optimization parameter, i.e. the asymptotic capacity $T_\infty^\star(g,\P,\Delta,\beta,\Lambda)$ 
(Definition \ref{def:asymptotic_capacity}).
Notice in Lemma \ref{lm:p_irsa_2lvl}, the relation between $p_i$ and $g$ is
non-convex and hence the objective function $T_\infty(g,\P,\Delta,\beta,\Lambda)= g(1 - \Lambda(p_\infty))$ is also
non-convex due to the composition of $\Lambda$ and $p_\infty$, and hence does
not admit any simple closed form expression.
A brute
force approach to maximize $T_\infty(g,\P,\Delta,\beta,\Lambda)$ can be used for this purpose, where the space of $g$ can be
discretized, and binary search like technique can be used to find
the optimal $g$. The finer the discretization, the finer the precision of the
approximation of $T_\infty^\star(\P,\Delta,\beta,\Lambda)$.

Next, we present a remark which will help us redefine the asymptotic capacity
$T_\infty^\star(\P,\Delta,\beta,\Lambda)$, and simplify the optimization of
asymptotic throughput for the case of IRSA-DPC.

\begin{remark}\label{rem:lossless}
We now make an important observation from our simulation results, but without proof. 
Notice that in Fig. \ref{subfig:2lvlcompare}, the capacity, as defined in Definition
\ref{def:asymptotic_capacity}, is achieved for a load $g$ such that $T_\infty^*
(\P,\Delta, \beta, \Lambda) =g$. After that point, there is a sudden drop in the throughput 
(one could call this a phase transition point). Therefore, we call the region $g
\leq T^\star_\infty$ as the \emph{lossless region}, and the region
$g > T^\star_\infty$ as the \emph{lossy region}.  We have noticed this
phenomenon in all our experiments with IRSA-PC, although a rigorous proof for
this remains elusive. However, we will use this fact in the further sections in
order to derive some theoretical results. Note that this phenomenon seems to be true only
for IRSA-PC and not for SA-DPC (as can be seen in Fig. 
\ref{subfig:2lvlcompare}). We believe that
this phenomenon occurs in schemes for which a DE analysis is accurate. DE
analysis is valid for IRSA-PC but not for SA-DPC.
\label{rm:max_throughput}
\end{remark}

\subsection{Finding asymptotic throughput $T_\infty^\star(\P,\Delta,\beta,\Lambda)$ for fixed $\Delta, \Lambda$}
As a consequence of Remark \ref{rm:max_throughput}, we can express the
asymptotic capacity $T_\infty^\star(\P,\Delta,\beta,\Lambda)$ as:
\begin{equation}
    T_\infty^\star(\mathcal{P}, \Delta, \beta,  \Lambda)
         = \sup \{g \geq 0 \text{ }| \text{ }\liminf_{i \to \infty} p_i = 0\},
         \label{eq:capacity}
\end{equation}
where the sequence $\left\{ p_i \right\}$ is computed with the DE equations as described in Lemma
\ref{lm:p_irsa_2lvl} with $p_1 = f_p(1)$. The constraint $\liminf_{i \to \infty}
p_i = 0$ follows from the observation made in Remark \ref{rm:max_throughput}
that the asymptotic capacity of a system is achieved at the maximum load $g$
such that $P_L(g,\P,\Delta,\beta,\Lambda) = 0$ which is equivalent to $\Lambda
(p_\infty) = 0$ which in turn is equivalent to the constraint $\liminf_{i \to
\infty} p_i = 0$. A necessary and
sufficient condition for the constraint $\liminf_{i \to \infty} p_i = 0$ to be
true is that $q > f_{q}(f_{p}(q))$
for every $q \in (0,1]$ (this is clearly a
sufficient condition as it leads to a decreasing sequence $q_{i}$, it
also turns out to be a necessary condition\cite{1795974}), where $f_p,f_q$ are
specified in Lemma \ref{lm:p_irsa_2lvl}. Since the necessary and sufficient
condition needs to hold for each $q \in (0,1]$, it is impractical to verify
this and hence the simplest way to check if $\liminf_{i \to \infty} p_i = 0$ is
to perform the DE using Lemma \ref{lm:p_irsa_2lvl} for large enough
iterations. {\bf This is why having Lemma \ref{lm:p_irsa_2lvl} is useful because
without the characterization of the DE equations, we can not compute the
asymptotic throughput $T_\infty(g,P,\Delta,\beta,\Lambda)$}.
From our simulations, we found that $500$ iterations of DE are sufficient for
$p_i$ to converge to some $\epsilon$-neighbourhood $\mathcal{N}_{\epsilon}
(p_\infty),\epsilon =10^{-10}$ around $p_\infty$.

\subsection{Finding asymptotic throughput $T_\infty^\star(\P,\beta)$ optimized over $\Delta, \Lambda$}
Optimizing $T_\infty^\star(\P,\Delta,\beta,\Lambda)$ with respect to
$g,\Lambda,\Delta = \{\delta,1 - \delta\}$, using the brute force approach is
practically infeasible because of exponential complexity. Therefore, we make use of empirical black box non-convex solvers. We now
formally present the optimization problem which we feed into our black box
non-convex solver.
\begin{equation}
\begin{aligned}
    & \underset{g,\lambda_2, \dots, \lambda_{\text{deg}},\delta}{
    \text{maximize}}
\qquad  g \\
& \text{subject to} 
&  \lambda_i \geq 0 , \; i = 2, \ldots, \text{deg},  \sum \lambda_{i} = 1, \\
& & q > f_{q}(f_{p}(q)) , \; \forall q \in (0,1], 1 \geq \delta \geq 0,
\end{aligned}
\label{eq:irsa_pc_optimization}
\end{equation}
where $\{\lambda_i\}_{i = 2}^{\text{deg}}$ are parameters of the
edge-perspective degree distribution of the graph defined
in Section \ref{sec:System-Model}, and $f_{q}(\cdot)$, $f_{p}(\cdot)$ are functions parameterized by DE in
Lemma \ref{lm:p_irsa_2lvl}. Note here that if we fix $\Lambda$ and $\delta$,
$\{\lambda_i\}$ in turn gets fixed and the solver essentially solves for
$T_{\infty}^\star(\P,\Delta,\beta,\Lambda)$.

We use the Differential Evolution algorithm to
solve \eqref{eq:irsa_pc_optimization} and use the \textit{MATLAB} code 
provided in \cite{diffevolution}. To run the
optimization program, we simply specify the optimization program as
stated in \eqref{eq:irsa_pc_optimization}.
Note that the last constraint
$q > f_{q}(f_{p}(q))$ needs to be applied for every point $q \in (0,1]$, which
leads to an uncountable
number of constraints. This is {\bf where Lemma \ref{lm:p_irsa_2lvl}} becomes useful, and we use the 
derived expressions for $f_{p}(q)$ and $f_{q}(p)$ in Lemma \ref{lm:p_irsa_2lvl}. To be specific, we 
pick a large number of points in the interval $ (0,1]$, and
apply the constraints $q > f_{q}(f_{p}(q))$ on only those points. 
Recall that the  maximum repetition
degree ``deg'' and the size of the set with powers available
$|\mathcal{P}|$ is the input to the optimization program
\eqref{eq:irsa_pc_optimization}, both of which are system parameters known
ahead of time, driven by constraints such as average power constraint etc.

Since the optimization problem \eqref{eq:irsa_pc_optimization} is
non-convex, the numerical solver's solution carries no guarantees with respect to the
global optimal capacity possible. Therefore, we supplement our results by deriving information theoretic throughput upper bounds next, and in section \ref{subsubsec:upper_bound_comparisons} we
show that the solution returned by the numerical solver to
\eqref{eq:irsa_pc_optimization} is close to the upper bounds for most of the
systems that we simulated.
	
\subsection{Upper bounds on the asymptotic throughput for IRSA-DPC}


In this section, we will derive upper bounds on the throughput of IRSA system with 2 power levels using Lemma \ref{lm:p_irsa_2lvl} and Remark
\ref{rm:max_throughput}. Since the DE equations of Lemma \ref{lm:p_irsa_2lvl} are
only accurate for systems with high multiplicative-gap $P_1/P_2$ in power levels, in order for Lemma \ref{lm:p_irsa_2lvl} to  be applicable for a general 
system with arbitrary power levels, we need the following lemma (Proof in Appendix \ref{app:lem:monotone}) that states that larger the ratio,
$P_1/P_2$, higher the throughput.

\begin{lemma}\label{lem:monotone}
    Consider two  communication systems following IRSA-DPC protocol: System 1 given by the parameters $\left\{ M, \P^1, \Delta,
    \beta, \Lambda \right\}$ 
    and system 2 given by $\left\{ M, \P^2, \Delta, \beta, \Lambda \right\}$. Let 
    $\P^1 = \left\{ P^1_1, P^1_2 \right\}$ and $\P^2=\left\{ P^2_1, P^2_2 \right\}$ denote the power levels 
    in the two systems respectively, with $\Delta = \left\{ \delta, 1-\delta \right\}$ being the
    common power-choice distribution for $\P^1$ and $\P^2$. Let $k_1 = P^1_1/\beta
    P^1_2$ and $k_2 = P^2_1/\beta P^2_2$. If $k_1 \geq k_2 \geq
    1$, then for a 
    given load $g$, $\quad \
            T(g, M, \P^1, \delta, \Lambda ) \geq T(g, M, \P^2, \delta,
        \Lambda)$.
\end{lemma}

Using Lemma \ref{lem:monotone}, we next derive an upper bound on the throughput as a function of  the average repetition rate $R$ of the
repetition scheme.
For the rest of this section, by IRSA-DPC system, we mean a system $\left\{ M, \P = \left\{ P_1, P_2 \right\}, \Delta =
    \left\{\delta, 1-\delta\right\}, \beta, \Lambda \right\}$.

\begin{theorem}[Upper Bound 1]
    Consider an IRSA-DPC system, where $R$ is the average repetition
    rate, $R = \sum_l l\Lambda_l$. With load $g$, if an asymptotic throughput
    $T_\infty$ is achievable on this system, then the following relation must be
    satisfied:
\begin{equation}
    \frac{(\delta^{2}-2)}{RT_\infty} + e^{-RT_\infty} \left( \frac{(1-\delta^2)}{RT_\infty} + \delta(1-\delta)\right) 
    + \frac{e^{-RT_\infty\delta}}{RT_\infty } + \frac{1}{R} \leq 0.
    \label{eq:upperbound_1_temp}
\end{equation}
\label{th:upperbound_1_temp}
\end{theorem}
\vspace{-0.4in}
\begin{proof}
      First, as a consequence of Lemma \ref{lem:monotone}, we can upper bound the throughput on
    the system $\left\{ M, \left\{ P_1, P_2 \right\}, \Delta, \beta, \Lambda \right\}$ by the throughput on another
    system $\left\{ M, \left\{ P'_1, P'_2 \right\}, \Delta , \beta, \Lambda \right\}$ with the
    power levels $P'_1$ and $P'_2$ being such that $P'_1/\beta P'_2$ is large enough for the DE equations in
    Lemma \ref{lm:p_irsa_2lvl} to be accurate.
So, for the rest of the proof, the analysis will be on the system $\left\{ M, \left\{ P'_1, P'_2
\right\}, \Delta , \beta, \Lambda \right\}$. The proof uses the DE equations: $q = f_{q}(p)$ and $p =
f_{p}(q)$. The DE can be visualized on an EXIT chart \cite{771431,957394} as shown in
Fig. \ref{fig:exit_chart}. 

Define the area between $f_{q}(p)$ as defined in \eqref{eq:q_iteration} and the
x-axis as $A_{q}$, and the area between $f_{p}(q)$ as defined in \eqref{eq:p_iteration}
and the y-axis as $A_{p}$: 
    $A_{q} = \int_{0}^{1} f_{q}(p)dp,\  \text{and} \ A_{p} = \int_{0}^{1} f_{p}(q)dq$.
To upper bound the asymptotic throughput $T_\infty(g,\P,\Delta,\beta, \Lambda)$ for a particular load $g$, it is
sufficient to upper bound the asymptotic capacity $T_\infty^{*}(\P,\Delta,\beta, \Lambda)$. 
As postulated in Remark \ref{rm:max_throughput}, the capacity
is achieved when the system is in the lossless region. In terms of the DE equations, the system is
in the lossless region implies that $p_i,q_i \rightarrow 0$ as $i \rightarrow \infty$. A necessary and sufficient 
condition for $p_i,q_i \rightarrow 0$ is that the two curves $f_{p}, f_{q}$ in the
EXIT chart in Fig. \ref{fig:exit_chart} are non-intersecting\cite{csa_paper}. Clearly, a
necessary condition for the two curves to be non-intersecting is that the area
covered by the two curves is not greater than 1: $A_{p} + A_{q} \leq 1$. The respective 
areas can be computed as:
\begin{equation*}
    A_{q} = \int_{0}^{1} \lambda(p)dp = \frac{1}{\Lambda'(1)} = \frac{1}{R}, \\
 \end{equation*}
    \begin{equation*}
    A_{p} = \int_{0}^{1} f_{p}(q)dq = \frac{\delta^{2}-2}{RT} + e^{-RT} \left(
    \frac{1-\delta^2}{RT} + \delta(1-\delta) \right) + \frac{e^{-RT\delta}}{RT} + 1.
\end{equation*}
Substituting this in $A_{p}+A_{q} \leq 1$ gives us the result.
\end{proof}

Theorem \ref{th:upperbound_1_temp} gives an upper bound for all repetition
distributions $\Lambda(x)$ such that their rate is $R$. In the next result, we derive an
upper bound for all $\Lambda(x)$'s, irrespective of their rates.

\begin{figure}[h]
    \centering
    \includegraphics[scale=0.6]{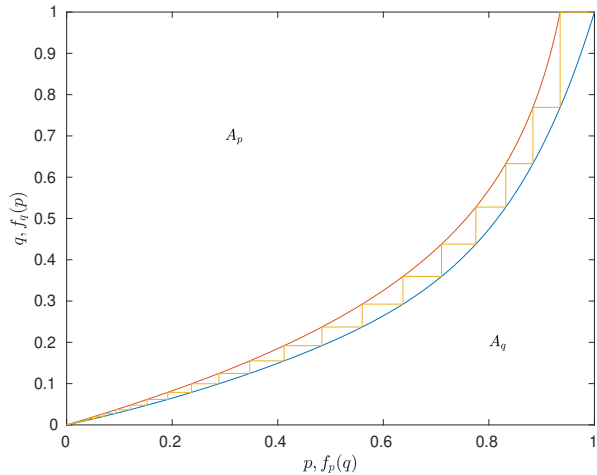}
    \caption{Exit Chart depiction of DE. $f_{q}(p)$ is plotted on the y-axis as
    a function of $p$ in the x-axis. $f_{p}(q)$ is plotted on the x-axis as
a function of $q$ in the y-axis. The DE recursion starts at the point (1,1),
and then proceeds by projections onto the $f_{q}(p)$ curve and $f_{p}(q)$ curve
sequentially. The point eventually converges to $(p_\infty, q_\infty)$. So, in 
the lossless region, the DE converges to (0,0) on the chart.
In the lossy region, the DE does not converge to (0,0).}
    \label{fig:exit_chart}
\end{figure}

\begin{corollary}[Rate Independent Upper Bound]\label{cor:riub}
    For an IRSA-DPC system, the asymptotic throughput
$T(g, M, \P , \Delta, \beta, \Lambda)$ for any
load $g$ satsifies $T_\infty \leq 2 -
\delta^2$.
\label{cor:rate_independent_ub}
\end{corollary}
\begin{proof}
The second and third term in the LHS of \eqref{eq:upperbound_1_temp} are
non-negative. Therefore, a necessary condition for
\eqref{eq:upperbound_1_temp} to be satisfied is: $\frac{(\delta^{2}-2)}{RT} +
 \frac{1}{R} \leq 0$, which gives us our result for all $R > 0$.
\end{proof}

\begin{corollary}[Upper Bound 2]
Let an asymptotic throughput of $T_\infty$ be achievable for a IRSA-DPC system. Let $l(p)$ be a line
which is tangent to the curve $f^{-1}_{p}(p)$ that passes through the point
$ (1,1)$. Let the point of contact of $l(q)$ and $f^{-1}_{p}(p)$ be
$ (p_{c},q_{c})$. Define $A_{min}$ as the area between the line $l(p)$ and the
curve $f^{-1}_{p}(p)$ from the range $ (p_c,q_c)$ to $ (1,1)$. Then, we have
$A_{p} + A_{q} + A_{min} \leq 1$.
\label{cor:upper_bound_3}
\end{corollary}

\begin{proof}
    As mentioned in the proof of Theorem \ref{th:upperbound_1_temp}, a necessary and sufficient condition for the density evolution of $p,q$ to converge to $0$ is that the two curves $f_{q}(p)$ and $f^{-1}_{p}(q)$ do not intersect. This implies that the curve $f_{q}(p)$ lies below the curve $f^{-1}_{p}(q)$. Therefore, $f_{q}(p_{c}) \leq q_{c}$. Also, observe that, $f_{q}(p) = \lambda(x)$, is a convex function (since, $\lambda(x)$ is a polynomial with non-negative coefficients). Hence, $f_{q}(p)$ lies below the line $l(x)$ (which is above the line connecting two points of $f_{q}(p)$). We have shown that the curve $f_{q}(p)$ cannot enter the area between the line $l(x)$ and $f^{-1}_{p}(x)$, which gives us our result.
\end{proof}

\begin{theorem}[Upper Bound 3]
    For an IRSA-DPC system, let the probability of a packet being repeated twice, $\Lambda_{2}$, be fixed. If an asymptotic
 throughput $T_\infty$ is achievable for some load $g$ on this system, then it
 satisfies $T_\infty < \min\big\{ 2-\delta^2, \frac{1}{2 \left(1 + 2\delta^2 -
 2\delta \right) \Lambda_2 } \big\}$.
\label{th:upper_bound_2}
\end{theorem}
\begin{proof}
Corollary \ref{cor:rate_independent_ub} includes communication schemes with arbitrary rates, which
clearly includes schemes with a particular fixed $\Lambda_2$ parameter. Therefore, $(2-\delta^2)$
is an upper bound on schemes with the fixed parameter $\Lambda_2$. The other term in the upper bound
expression is obtained by upper bounding the capacity as derived next.
Again, as a consequence of Lemma \ref{lem:monotone}, we can derive an upper bound using
the DE analysis from Lemma \ref{lm:p_irsa_2lvl}. Combining
\eqref{eq:p_iteration} and \eqref{eq:q_iteration}, the DE recursion is written as  $q_{i+1} =
f_{q}(f_{p}(q_{i}))$. As mentioned in Remark \ref{rm:max_throughput}, the capacity
is equal to the maximum load such that the system is still in the
lossless region. Therefore, the DE recursion at the capacity achieving load $g$
has to be such that $q > f_{q}(f_{p}(q))$
for every $q \in
(0,1]$. In the region where $q \rightarrow 0$, by taking the partial derivative
of $f_{q}(f_p(q))$ with respect to $q$ ,this condition becomes
equivalent to $\frac{\partial f_{q}(f_{p}(q))}{\partial q} \Bigg|_{q=0} < 1$.
This expression can be evaluated to obtain the result.

\begin{align*}
    \frac{\partial f_{p}(q)}{\partial q} \Bigg|_{q=0} &=  RT_\infty(1-\delta) +
        RT_\infty\delta^{2} - RT_\infty\delta(1-\delta), \quad
    \frac{\partial f_{q}(p)}{\partial p} \Bigg|_{p=0} = \lambda_{2} \\
    \frac{\partial f_{q}(f_{p}(q))}{\partial q} \Bigg|_{q=0} &= \frac{\partial f_{p}(q)}{\partial q} \Bigg|_{q=0}\frac{\partial f_{q}(p)}{\partial p} \Bigg|_{p=0} = \lambda_{2} \left( RT_\infty(1-\delta) +
        RT_\infty\delta^{2} - RT_\infty\delta(1-\delta) \right) \\
        &\stackrel{(a)}{=} 2\Lambda_{2}T_\infty \left( 1 + 2\delta^{2} - 2\delta
        \right), 
\end{align*}
where (a) is obtained by using the relation $R\lambda_{2} = 2\Lambda_{2}$ and rearranging the terms. Now, using the inequality, $\frac{\partial f_{q}(f_{p}(q))}{\partial q}\Bigg|_{q=0} < 1$, completes the proof.

\end{proof}
\subsubsection{Upper Bound Comparisons}
\label{subsubsec:upper_bound_comparisons}
To compare the different upper bounds, we choose $\P = \left\{ 10, 1 \right\}$, and the capture threshold $\beta = 2$.
The power choice distribution is $\Delta = \left\{ \delta, 1- \delta \right\}$, where we will present the
results for the different values of $\delta$. 
Since the different upper bounds deal with different parameters being fixed, we
group the upper bounds and the parameters they depend on with the same column
color in Table I.

\begin{table}[h]
    \centering
\begin{tabular}{|c|c|>{\columncolor[gray]{0.9}}c|c|>{\columncolor[gray]{0.9}}c|>
{\columncolor[gray]{0.9}}c|>{\columncolor[gray]{0.7}}c|>{\columncolor[gray]
{0.7}}c|c|}
\hline
\textbf{$\Lambda(x)$}   & \textbf{$\delta$} & $R$ & \textbf{$T^{\star}_\infty$} 
& \textbf{UB1} &
 \textbf{UB2} & $\Lambda_2$ & \textbf{UB3} & \textbf{Rate Ind. UB}  \\ \hline
 $\Lambda^{\text{liva}}(x) = 0.5x^2 + 0.28x^3 + 0.22x^8$ & 1 & 3.6                 & 0.938            & 0.9695        
       & 0.962  &  0.5 & 1  & 1                       \\ \hline
       $\Lambda^{\text{liva}}(x) = 0.5x^2 + 0.28x^3 + 0.22x^8$ & 0.4 & 3.6              & 1.67             & 1.756    
       & 1.738   &  0.5  &  1.84   & 1.84            \\ \hline
       $\Lambda^{1}(x) = 0.56x^2 + 0.21x^3 + 0.23x^8$           & 0.4  & 3.59             & 1.67 & 1.756
       & 1.734   & 0.56 & 1.717  & 1.84     \\ \hline
       $\Lambda^{2}(x) = 0.6x^2 + 0.2x^3 + 0.2x^8$           & 0.6  & 3.4             & 1.517 & 1.589
       & 1.579  & 0.6 & 1.581    & 1.64            \\ \hline
\end{tabular}
\caption{Comparisons of different Upper Bounds. 
 }
\label{tab:ub1_ub3}
\end{table}
\vspace{-.45in}

\begin{figure}
    \centering
    \includegraphics[scale=0.5]{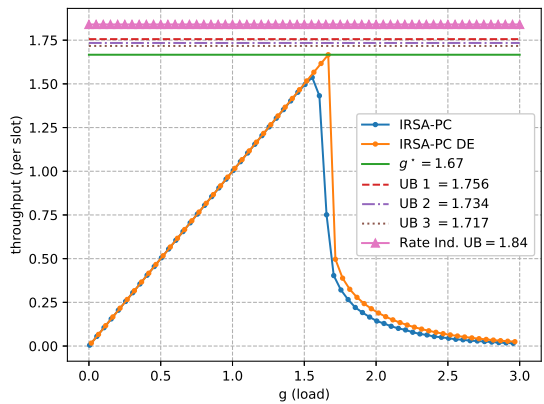}
    \caption{Comparing the various upper bounds on an IRSA-DPC system.
        The system is  $\left\{ M, \P = \left\{ 10,1 \right\}, \Delta = \{0.6,0.4\}, \beta = 2,
        \Lambda^1(x) \right\}$, where 
  $\Lambda^{1}(x) = 0.56x^2 + 0.21x^3 + 0.23x^8$ }
    \label{fig:upper_bound}
\end{figure}

\section{Numerical Simulations and Comparisons}\label{sec:sim}
In this section, we provide exhaustive simulation results to
supplement our analytical results. For all experiments in this section, unless
otherwise mentioned, the capture threshold is
$\beta=2$ and the number of slots is $M=1000$. For all our results presented
below, we will use $\Lambda(x) = \Lambda^{\text{liva}}(x) = 0.5x^2 + 0.28x^3 +
0.22x^8$, unless stated otherwise. For the case of dual power control in SA-DPC
and IRSA-DPC, we consider $\P = \{10P,P\}, \Delta = \{\delta = 0.4, 1- \delta
= 0.6\}$, where $P$ is the minimum transmitted power required such
that the packet is decodable at the receiver without any interference. The
choice of $\delta$ is motivated by the fact that $\delta = 0.4$ is the optimal
for IRSA-DPC for the $\Lambda^{\text{liva}}$ distribution (found using simulations). For the
case of $n = 3$ power
levels in IRSA-3PC, we consider $\P = \{100P,10P,P\},\Delta = 
\{0.27,0.39,0.34\}$.
$P_{\text{avg}}$ denotes the average power spent by the user in transmitting the
packets in a particular frame. We have that $P_{\text{avg}} = R \sum_{i = 1}^
{|\P|}P_i \delta_i$, where $R$ is the average repetition rate defined in Section
\ref{sec:System-Model}. For all our
results presented in this section, we present the throughput averaged over 100
iterations for each load $g$. 

In Tables \ref{tab:repetition_rate_avg_power} and III, we give a summary of throughput improvement with our proposed strategy, where Capacity (Sim) means the chosen value of $M$, while for Capacity (Asym) $M \to \infty$. 

\begin{table}
    \centering
\begin{tabular}{|c|c|c|c|c|c|}
\hline
Scheme  & $R$ & $\sum_{i = 1}^{|\P|}P_i\delta_i$& $P_{\text{avg}}$ & Capacity 
(Sim) & Capacity (Asym.) \\ \hline
 SA & $1$ & $P$ & $P$ & $0.367$ & $0.367$ \\ \hline
 SA-DPC & $1$ & $4.6P$ & $4.6P$ & $0.624$ & $0.657$ \\ \hline
 IRSA &  $3.6$ & $P$ & $3.6P$  & $0.841$ & $0.916$       \\ \hline
 IRSA-DPC & $3.6$ & $4.6P$ & $16.56P$ & $1.551$ & $1.667$           \\ \hline
 IRSA-3PC & $3.6$ & $31.24P$ & $112.46P$ & $1.941$ & $2.016$       \\ \hline
\end{tabular}
\caption{Repetition rate $R$ and average power $P_{\text{avg}}$ of different
schemes.}
    \label{tab:repetition_rate_avg_power}
\end{table}

Along with the simulations we also plot the DE curve 
(suffixed by DE in the legends) to
show the performance of the schemes in the asymptotic setting $ (M \to \infty)$.



\begin{figure}
\label{fig:simulations_comparisons}
\centering
\begin{subfigure}{0.48 \textwidth}
\centering
\includegraphics[width = \linewidth]{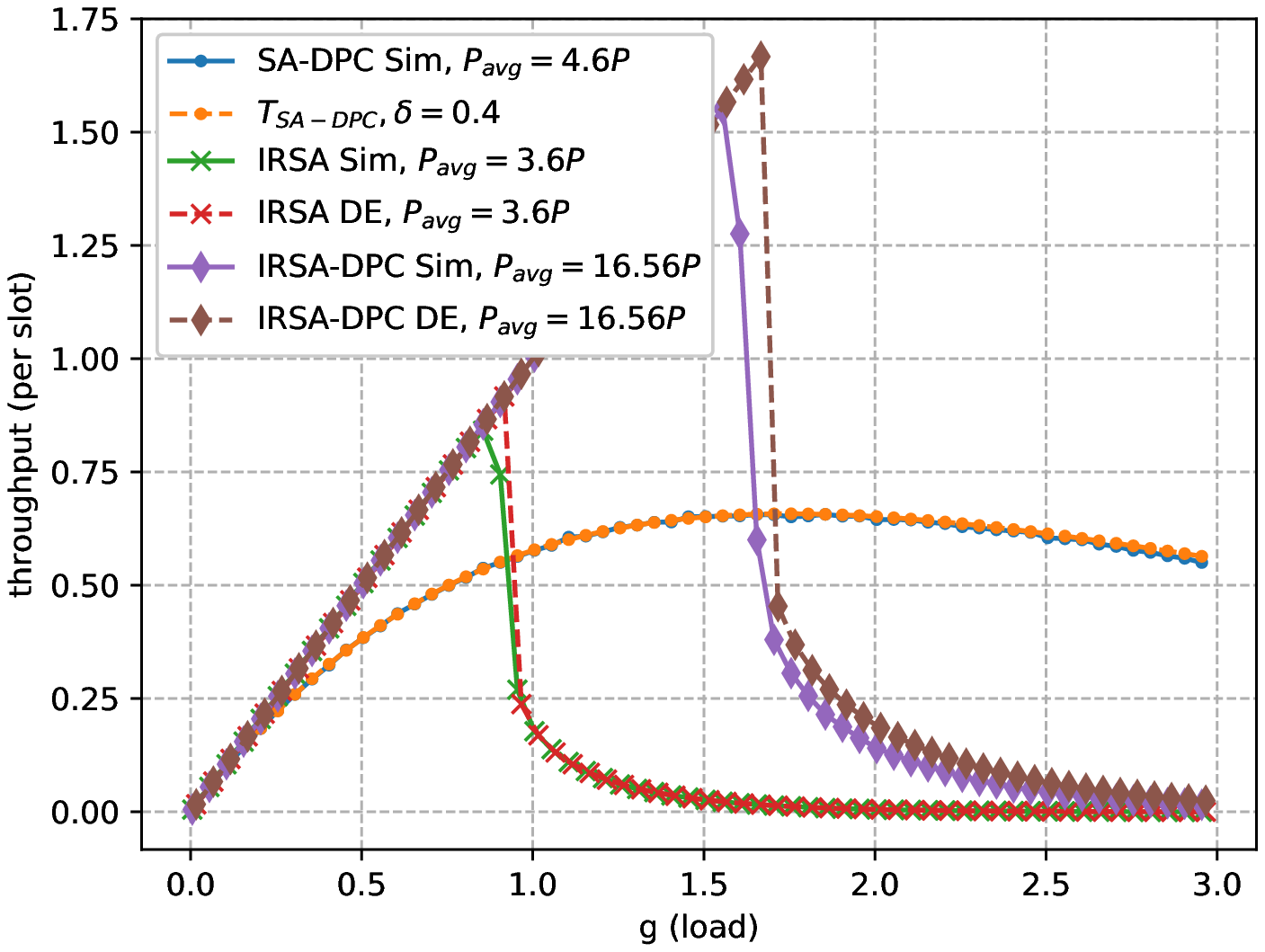}
\caption{}
\label{subfig:2lvlcompare}
\end{subfigure}
\begin{subfigure}{0.48 \textwidth}
\centering
\includegraphics[width = \linewidth]{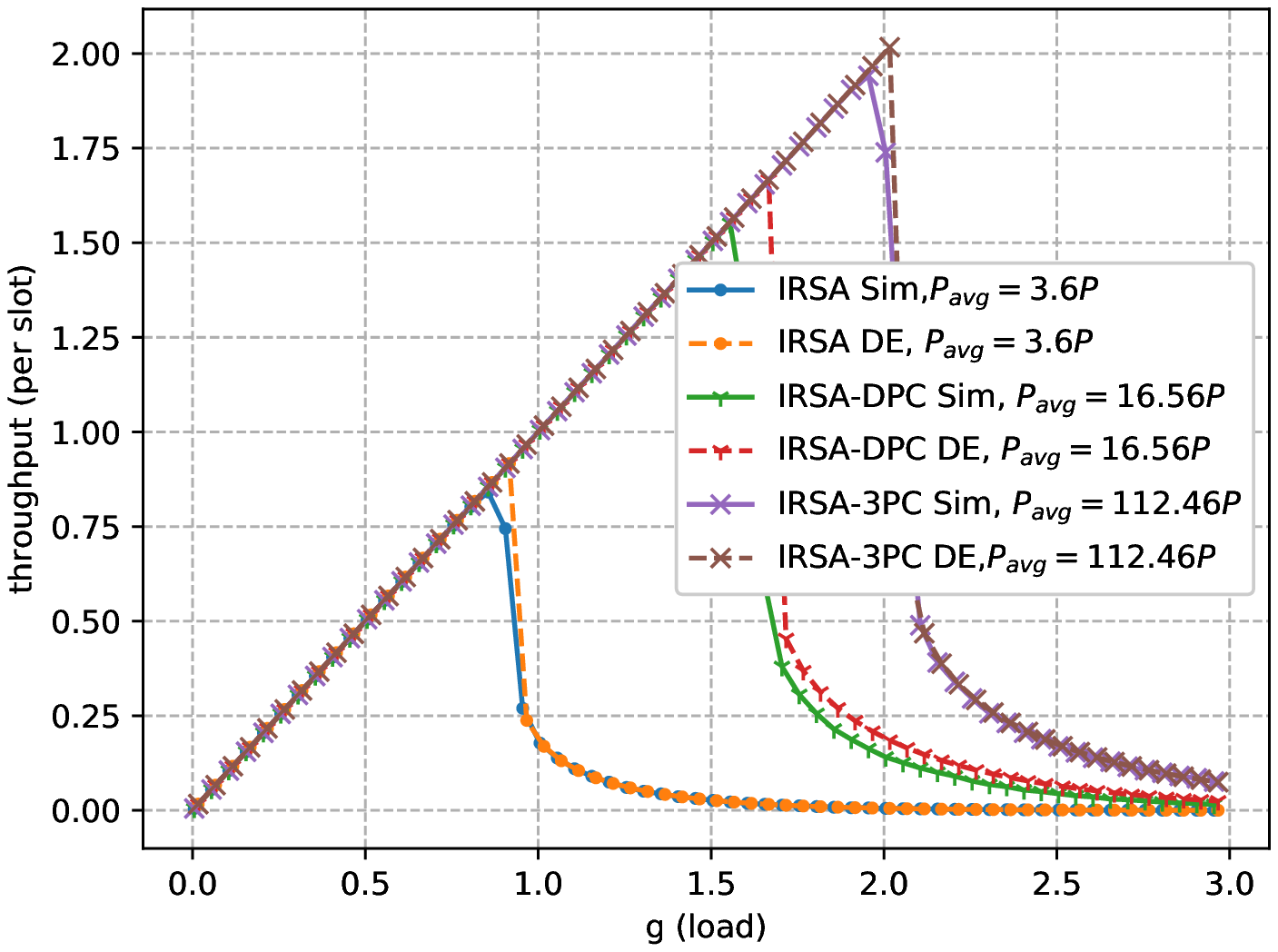}
\caption{}
\label{subfig:nlvlcompare}
\end{subfigure}
\begin{subfigure}{0.48 \textwidth}
\centering
\includegraphics[width = \linewidth]{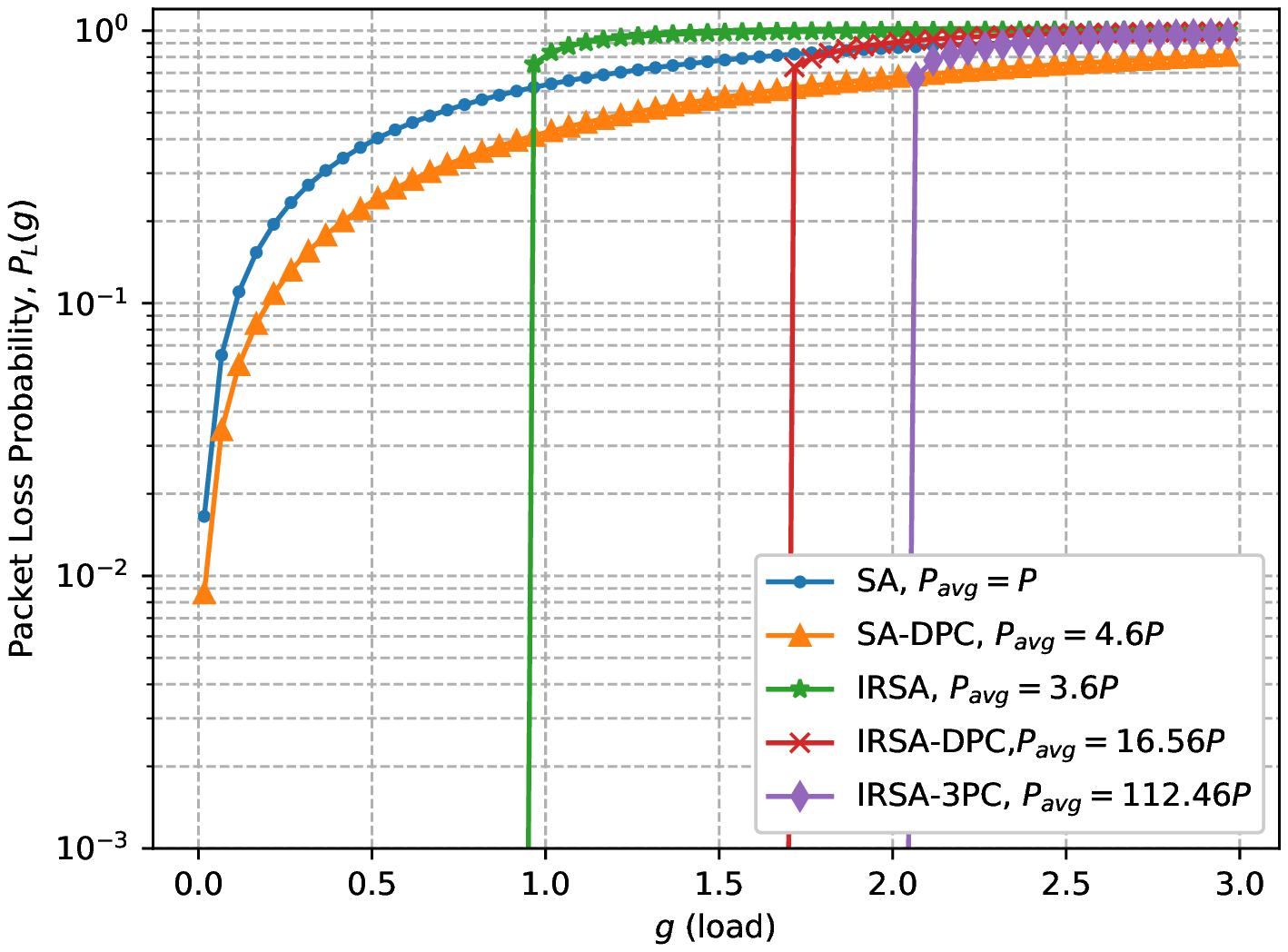}
\caption{}
\label{subfig:packet_loss_prob}
\end{subfigure}
\begin{subfigure}{0.48 \textwidth}
    \centering
    \includegraphics[width = \linewidth]{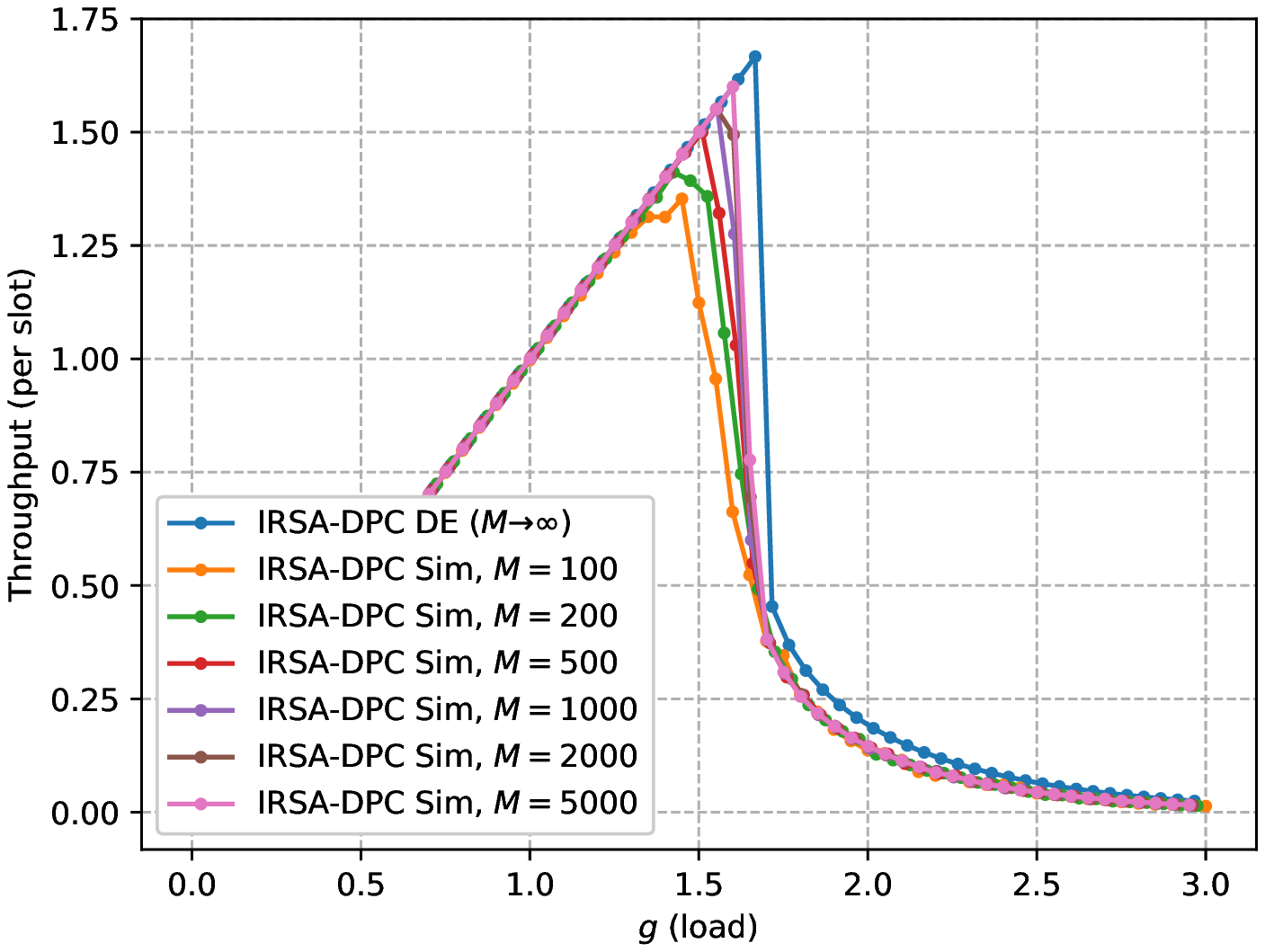}
    \caption{}
    \label{subfig:varyM}
\end{subfigure}
\begin{subfigure}{0.48 \textwidth}
    \centering
    \includegraphics[width = \linewidth]{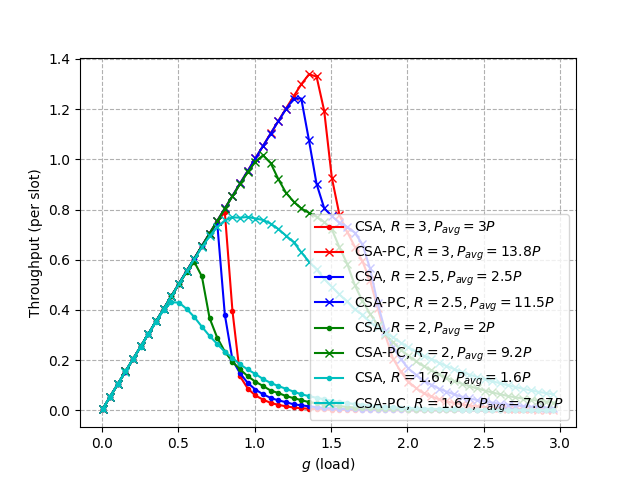}
    \caption{}
    \label{subfig:csa_csa_pc_multiple_rates}
\end{subfigure}
\begin{subfigure}{0.48 \textwidth}
    \centering
    \includegraphics[width = \linewidth]{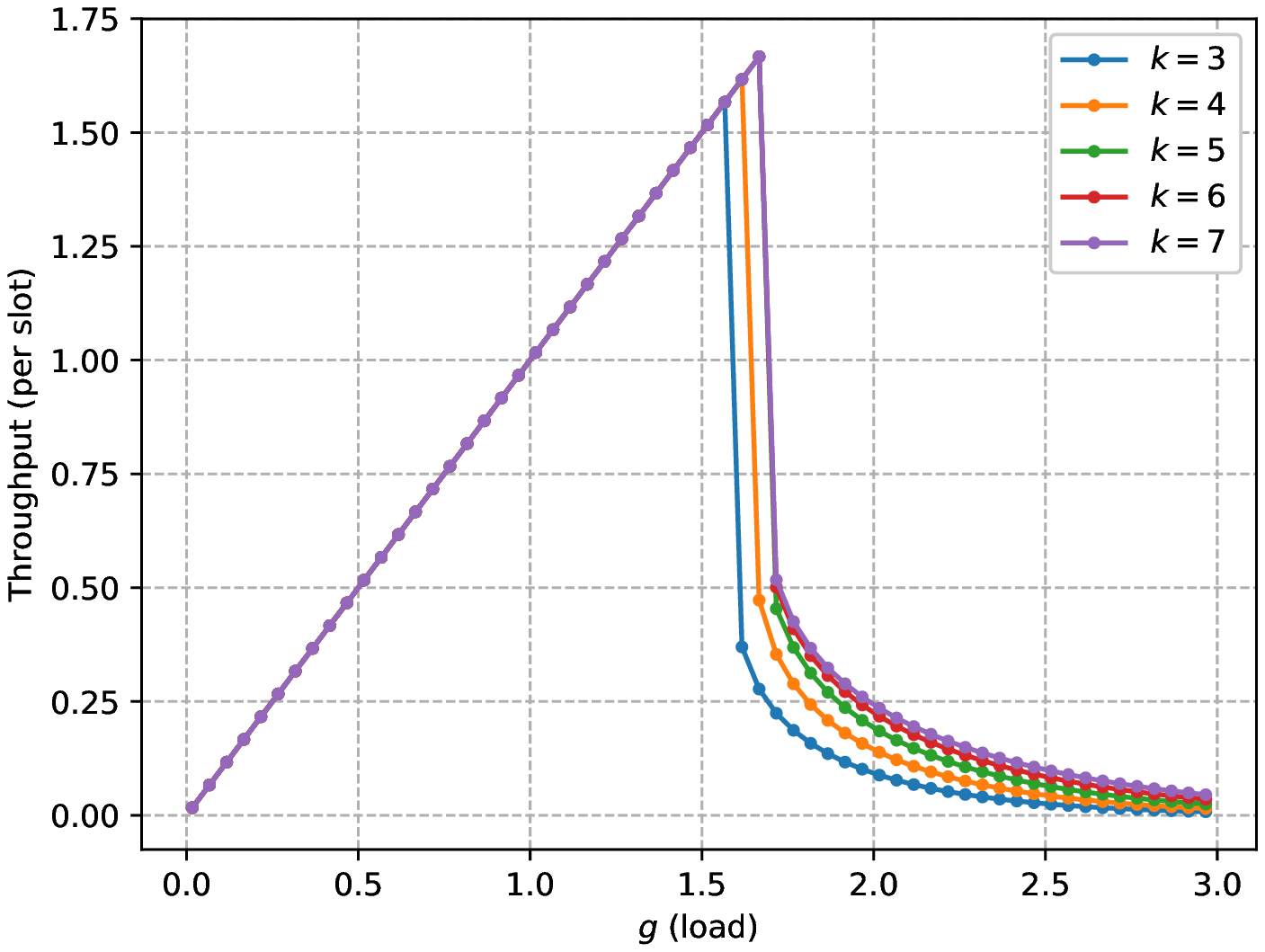}
    \caption{}
    \label{subfig:k_approximation}
\end{subfigure}
\caption{Throughput comparisons with description provided in Section \ref{sec:sim}.}
\end{figure}

Fig. \ref{subfig:2lvlcompare} highlights two types of gains in throughput due to: 
(a) \emph{coding} and (b) \emph{power multiplexing}. 
By comparing SA-DPC with IRSA-DPC, we can see the gains in throughput due
to the coding scheme that we employ. For instance, in the case of
SA-DPC and IRSA-DPC, the maximum throughput increases from $0.68$ (SA-DPC) to
$1.67$ (IRSA-DPC), accounting for $\sim 2.5 \times$ gains in throughput.
Similarly, comparing IRSA with IRSA-DPC, highlights the gains in throughput due to
power multiplexing. The maximum throughput increases from $0.92$ (IRSA) to $1.67$
(IRSA-DPC), accounting for $\sim 1.8 \times$ gains in throughput. It is also
important to note here that these gains in throughput come at the cost of 
increased average power $P_{\text{avg}}$.

Fig. \ref{subfig:nlvlcompare} compares the throughput of IRSA-nPC protocol with $n$ power levels for
$n = 1$ (IRSA), $n = 2$ (IRSA-DPC) and $n = 3$ (IRSA-3PC). Fig. \ref{subfig:nlvlcompare}
substantiates the intuition that as the number of power levels $n$ increases, so
does the maximum throughput, albeit at the cost of increased average power.
Hence there is a trade-off between the maximum achievable throughput and $P_
{\text{avg}}$ and the choice of using one protocol over the other would be subject
to other design constraints and parameters. 
It can be inferred from Fig. \ref{subfig:nlvlcompare} that as $n$
increases, the proportional gain in throughput decreases.

Fig. \ref{subfig:packet_loss_prob} highlights the packet loss probability $P_L
(g,\P,\Delta,\beta,\Lambda)$  (Definition \eqref{def:pl}) for SA,
SA-DPC, IRSA, IRSA-DPC and IRSA-3PC. For IRSA, IRSA-DPC, IRSA-3PC, the
packet loss probability $P_L$, i.e. the probability that a users' packet is not
decodable is $\approx 0$ for loads $g$ less than the respective capacities of
each of the protocols. Under these loads, i.e. $g$ less than the respective capacities, we
have lossless transmission which we alluded to in Remark 
\ref{rm:max_throughput}. At the respective capacities of IRSA, IRSA-DPC and
IRSA-3PC, we notice a phase transition where till we approach the capacities, the packet
loss probability is $\approx 0$ and then suddenly increases and keeps increasing
as the load $g$ increases beyond the respective capacities. One interesting
point to note here is the fact that as the
load $g$ increases, we see that SA-DPC has the least packet loss probability and
therefore highest throughput as the load $g$ increases. Therefore for systems
which need to support large loads $g$, SA-DPC is superior in comparison to the
other considered protocols.

Fig. \ref{subfig:varyM} demonstrates the approximation of the simulated
throughput in comparison to the asymptotic throughput obtained using the DE
equations. The DE curve corresponds to the case of $M \to \infty$ and the
simulations for some finite $M$. This highlights that as $M$ increases, the
approximation approaches the asymptotic throughput.

Fig. \ref{subfig:csa_csa_pc_multiple_rates} presents the simulation results for
CSA and CSA-DPC for different rates. CSA stands for Coded Slotted Aloha --
introduced in \cite{csa_paper}, which is an extension of IRSA. Contrary to IRSA, 
wherein the user packets are simply repeated, in CSA, the user packets
are encoded (using some linear codes) prior to the transmission in the MAC frame. 
The encoding is done through local component codes, drawn independently by the
users from a set of
component codes $\{\mathscr{C}_h\}_{h = 1}^{\theta}$. The p.m.f over these
component codes is denoted by $\{\Lambda_h\}_{h = 1}^{\theta}$. 
 We define the rate as $R := \frac{\sum_
{h=1}^\theta\Lambda_h n_h}{k}$ where $\theta$ is the total number of codes considered,
$\Lambda_h$ is the probability of choosing component code $\mathscr{C}_h$ whose
length is $n_h$. Each component code has a common dimension $k$. The user
packets are decoded on the receiver side by combining SIC with local component
code decoding to recover from collisions. CSA-DPC is the dual power control
variant of CSA, which is analogous to the IRSA-DPC protocol, in which each bit in
the MAC frame is transmitted with a power level drawn from the set $\{P_1,P_2\}$
with p.m.f. $\{\delta,1 - \delta\}$.
Since theoretically analysing power control for CSA-DPC protocol is
intractable, we present some empirical results using simulations comparing the
the throughput obtained using CSA and the dual power control CSA-DPC
protocol for different rates. The theoretical intractibility is due to lack of a
closed form expression for the density evolution equations
for the CSA-DPC protocol. 
We consider
rates $R = 3, 5/2, 2, 5/3$. For the component codes and the
respective probability distribution $\Lambda$ for $k= 2$, refer to Table II of 
\cite{csa_paper}.
The increase in
capacity in the case of CSA and CSA-PC for different rates are summarized in Table III.

Fig. \ref{subfig:k_approximation} illustrates the practical implication of
assuming $k$ being large enough, recall Lemma \ref{lm:k_to_inf_approx}. In Remark
\ref{rmk:k_to_inf_approx} we pointed out
that for practical purposes, $k=5$ can be considered to be a large enough $k$. Fig. \ref{subfig:k_approximation}
illustrates that point with reference to IRSA-PC. The system is $\{M, \P = 
\{k\beta P_2,P_2 \}, \Delta = \{0.4, 0.6 \},
\beta\}$, and the IRSA-PC protocol is $\Lambda(x) = 0.5x^2 +
0.28x^3 + 0.22 x^8$. One can observe that the throughput
curve converges as $k$ is made to grow from $k=3$ to $k=7$, and that the curves are virtually
identical beyond $k=5$. We have noticed this phenomenon in all our experiments for a range of
parameters of IRSA-PC.
\begin{table}[!h]\label{tab:csa}
\begin{center}
\begin{tabular}{| c | c | c | }
\hline
$R$ & CSA Capacity (Sim) & CSA-DPC Capacity (Sim)\\
\hline
$3$ & $0.789$ & $1.339$ \\
\hline
$5/2$ & $0.748$ & $1.242$ \\
\hline
$2$ & $0.592$ & $1.016$ \\
\hline
$5/3$ & $0.431$ & $0.771$ \\
\hline
\end{tabular}
\end{center}

\caption{Comparison of throughput with CSA and CSA-PC.}
\end{table}

\section{Path Loss Model}
\label{appendix:path-loss-model}
In prior work \cite{liva}, theoretical analysis for the throughput of the IRSA protocol without any power control has been carried out 
only for the ideal channel model, discussed in Section \ref{sec:System-Model}. 
A more realistic channel model called -- the {\it path-loss channel} model is parameterized by two parameters; $d_
{\text{min}}$ and the path-loss exponent
$\alpha$. With the {\it path-loss channel}, a base-station is the receiver, and the users are assumed to be
distributed according to some spatial distribution around the base-station in a circle of radius $r$. If a user
transmits a packet with a power $P$ from a distance of $d$, then the received
power is given by:

\begin{equation}
    P_{rec} = 
    \begin{cases}
	    P & d \leq d_{\text{min}}, \\ \label{pathlossmodel}
	    P \left( \frac{d}{d_{\text{min}}} \right)^{-\alpha} & d >
	        d_{\text{min}}.
    \end{cases}
\end{equation}

The goal of this section is to extend the approximate throughput guarantees of
the IRSA-nPC protocol with the ideal channel model to the  
path loss model when only a single power level is used by each node. This is  a fairly
challenging problem, and to
the best of our knowledge, \cite{near_far} is the only work in which the
authors have attempted to study the vanilla IRSA protocol (with a single power
level) with the path loss model, however, no closed form or explicit expression
for the density evolution
equations required for asymptotic analysis is provided there. 
Instead, the authors of 
\cite{near_far} use Monte-Carlo simulations to estimate the coefficients in the
density evolution equations. Therefore there is no principled way of computing
or approximating the throughput of the vanilla IRSA scheme under the path loss
model. 

We make a connection between the  IRSA protocol under the path loss model with a single power level and the IRSA-nPC under the ideal channel model, and leverage the results 
developed in earlier sections to approximate the throughput of  the IRSA protocol under the path loss model.
We make an additional assumption, that the knowledge of the spatial distribution of users is known.


Before going into the
details of this approximation framework, it is instructive to understand in an
intuitive manner as to why it is difficult to analyse the IRSA protocol under the
path loss model. The heart of the analysis of the IRSA scheme is the capture
effect which
essentially says that packet of user $i$ in slot $m$ can be decoded as long as
$SIR^{(i,m)} = \frac{P_{rec}^{(i)}}{\sum_{n \in \mathcal{R}^{0}_{j}\setminus i}
P_{rec}^{(n)}} \geq \beta$. Assuming that the radial distance of all the users
in more than $d_{\min}$, the expression for $SIR^{(i,m)}$ under the path loss
model assumption is $SIR^{(i,m)} = \frac{P\left(\frac{r_i}{d_{\min}}\right)^
{-\alpha}}{\sum_{n \in \mathcal{R}^{0}_{j}\setminus i}P\left(\frac{r_n}
{d_{\min}}\right)^{-\alpha}}$ which can be simplified to $SIR^{(i,m)} = 
\frac{r_i^{-\alpha}}{\sum_{n \in \mathcal{R}^{0}_{j}\setminus i}
r_n^{-\alpha}}$ where $r_j$ is the radial distance of user $j$ from the
basestation. 

Assuming the users are distributed randomly around the basestation,
$r_n$'s are random variables, hence the
denominator is  a random sum of random variables $\sum_{n \in 
\mathcal{R}^{0}_{j}\setminus i} r_n^{-\alpha}$. Unlike \cite{capture_fading},
which showed that SIR follows the Pareto distribution under the Rayleigh fading
channel model, we are unable to characterize the distribution of the SIR in the path loss model. 
Hence, the need to approximate
the path loss model with a $n$-power level IRSA-PC model.
To do this, we show that the vanilla IRSA scheme under the path loss model, and
the IRSA-nPC shceme under the ideal channel assumption are
\emph{approximately} equivalent. 

\subsection{Discretization For Approximation}
Let $P_{\min}$ be the minimum power received at the receiver (or basestation)
that can be successfully decoded. If $P_{\min} >0$ implies there exists a
distance $d_{\max}$
such that a user can be decoded if and only if the distance of the user from the
basestation is less than $d_{\max}$.
In our framework, the first step is to discretize the set of possible
received powers due to path loss which is $[P_{\min}, P]$. Note that $P$ is the
common power with which all the users transmit their packets and if the
distance of the user from the basestation is less than $d_{\min}$, then the
received power at the base station is $P$. We can discretize this continuous set
into $n$ power levels where $n = \floor*{\frac{\log P/P_
{\min}}{\log k\beta}} + 1$ such that $P_i = k\beta P_{i+1}, P_1 = P, P_n = P_
{\min}$. This discretized set with powers is the same set $\mathcal{P}$ (by
construction) which is considered in the IRSA-nPC model. 
In the IRSA-nPC model apart from the power level, we also considered the
distribution $\Delta$ from which these power levels are sampled
independently by the users. In the path loss model, since all the users transmit
the packets with a
common power $P$, the variation in the power of the received packets arises due
to their respective distances from the base station. The spatial distribution of
the users around the basestation determines the $\Delta$ distribution. 
Note that in the path loss model, the received power $P_{
\text{rec}}^i$ and the radial distance of the user from the receiver $r_i$ are
related. Therefore we can discretize the distance set of user $[0, d_
{\max}]$ similar to how we discretized the set of received powers $[P_{\min},
P]$. For each $P_i \in \mathcal{P}, d_i = \left(\frac{P}{P_{i}}\right)^
{1/\alpha}d_{\min}$. Define $d_0 := - d_1, d_{n + 1} := d_n = d_
{\max}$. Let $N_i = |\{j \in U : \frac{d_i + d_{i-1}}{2} \leq r_j \leq 
\frac{d_i + d_{i+1}}{2}\}|$ denote the number of users whose radial distance
from the basestation lies between $\frac{d_i + d_{i-1}}{2}$ and $\frac{d_i + d_
{i+1}}{2}$. This gives us the $\Delta = \{\delta_i\}_{i}^n$ distribution as $\delta_i = 
\frac{N_i}{N}$ where $N_i$ is number which depends on how the users are
spatially distributied around the base-station.

\subsection{Case Study}
The best way to understand this framework
is with the help of an example. 
For this example let us assume that $P_{\min} = 0.01 P$.
Also make the additional assumption that $\beta = 2$ and the assumption of $k =
5$ follows from Remark \ref{rmk:k_to_inf_approx}.
Then it follows from the argument above that $d_{\max} \approx 4.64 d_{\min}$
and we can discretize this path loss model in $n = \floor*{\frac{\log 100}{\log
10}} + 1 = 3$ levels where $P_1 = P, P_2 = 0.1 P, P_3 = 0.01 P = P_{\min}$.

Let $L$ be the random variable denoting the position of a user relative to the
base station. In terms of the polar coordinates this random variable $L$ can be
characterized by two random variables $R_L$ and $\Theta_L$ which denote the
radial
distance of the user from the base station and the angle from a referenced axes
respectively.

Hence we can define $L := (R_L,\Theta_L)$. In our example let $R_L \sim 
\text{Unif}[0,d_{\max}],\Theta_L \sim \text{Unif}[0,2\pi]$ and assume that $R_L$
is independent of $\Theta_L$. Since $R_L$ is independent of $\Theta_L$, it
implies that the joint PDF of the position random variable $L$ can be written as
$f_L(r,\theta) = \frac{1}{2\pi d_{\max}}$. The reason for choosing such a simple assumption on the
spacial distribution of users is that the calculation of $\{\delta_k\}$ is
greatly simplified. Without loss of generality, we can assume $d_{\min} = 1$,
then we have that $d_1 = d_{\min} = 1, d_2 = \left(\frac{P}{0.1 P}\right)^{
\frac{1}{3}} \times d_{\min} = 2.15, d_3 = d_{\max} = 4.64$. Note that while
calculating $\{\delta_k\}$, for this example, since the radial distance of the
user is independent of the angle the user makes from a referenced axis. Also
since the radial distance of the user is uniformly distributed between $0$ and
$d_{\max}$, it follows that $\delta_i = \frac{\frac{d_i + d_{i+1}}{2} - \frac{d_
{i} + d_{i-1}}{2}}{d_{\max}} = \frac{d_{i+1} - d_{i-1}}{2 d_\max}$. Since $d_0 =
-d_1, d_{4} = d_3$, we have that $\delta_1 = \frac{d_1 + d_2}{2 d_{\max}}
\approx 0.34, \delta_2 = \frac{d_3 - d_1}{2 d_{\max}} \approx 0.39 , \delta_3 = 
\frac{d_3 - d_2}{2 d_{\max}} \approx 0.27$.

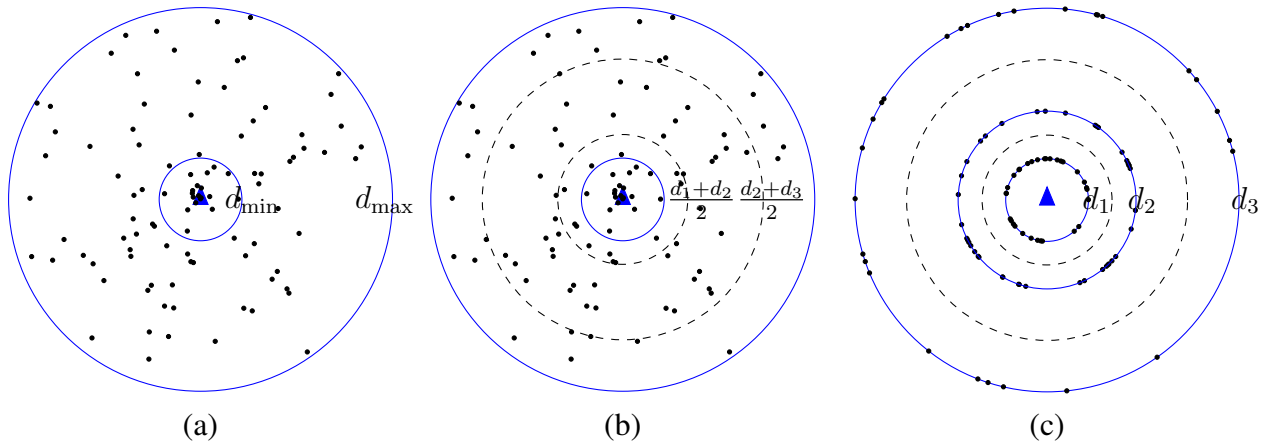
\begin{figure}[h]
\minipage{0.32\textwidth}
  \begin{tikzpicture}[scale = 0.55]
        \node[draw=blue,fill=blue,isosceles triangle,shape border
        rotate=90,minimum width=0.2cm,minimum height=0.2cm,inner sep=0pt] at (0,0) {};
        \draw[blue] (0,0) circle (1 cm);
        \node at (1.25, 0) (a) {$d_\min$};
        \draw[blue] (0,0) circle (4.64 cm);
        \node at (4.45,0) (c) {$d_{\max}$};
        \foreach \p in \points{
	        \draw[fill = black] \p circle (0.05 cm);
	    }
	   \node at (0,-5.5) (a) {(a)};
     \end{tikzpicture}
\endminipage \hfill
\minipage{0.32\textwidth}
 \begin{tikzpicture}[scale = 0.55]
        \node[draw=blue,fill=blue,isosceles triangle,shape border
        rotate=90,minimum width=0.2cm,minimum height=0.2cm,inner sep=0pt] at (0,0) {};
        \draw[blue] (0,0) circle (1 cm);
        \draw[blue] (0,0) circle (4.64 cm);
        \foreach \p in \points{
	        \draw[fill = black] \p circle (0.05 cm);
	    }
        \draw[dashed] (0,0) circle (1.57 cm);
        \node at (1.9,0) (r1) {$\frac{d_1 + d_2}{2}$};
        \draw[dashed] (0,0) circle (3.395 cm);
        \node at (3.6,0) (r2) {$\frac{d_2 + d_3}{2}$};
        \node at (0,-5.5) (b) {(b)};
	 \end{tikzpicture}
\endminipage\hfill
\minipage{0.32\textwidth}%
  \begin{tikzpicture}[scale = 0.55]
        \node[draw=blue,fill=blue,isosceles triangle,shape border
        rotate=90,minimum width=0.2cm,minimum height=0.2cm,inner sep=0pt] at (0,0) {};
        \draw[blue] (0,0) circle (1 cm);
        \node at (1.2,0) (d1) {$d_1$};
        \draw[blue] (0,0) circle (4.64 cm);
        \node at (4.8,0) (d3) {$d_3$};
        \foreach \p in \dispoints{
	        \draw[fill = black] \p circle (0.05 cm);
	    }
        \draw[dashed] (0,0) circle (1.57 cm);
        \draw[dashed] (0,0) circle (3.395 cm);
        \draw[blue] (0,0) circle (2.15 cm);
        \node at (2.3,0) (d2) {$d_2$};
        \node at (0,-5.5) (c) {(c)};
	 \end{tikzpicture}
\endminipage
\caption{Sub-Fig. (a) describes a random distribution of the users around the
basestation where the radius and angle with respect to a referenced axis are
sample uniformly and independently from Unif$[0,d_{\max}]$ and Unif$[0,2\pi]$
respectively. In sub-Fig. (b), the circular region around the base-station is
divided into zones. Users whose radial distance from the base-station is between
$0$ and $\frac{d_1 + d_2}{2} \approx 1.57$ belong to zone 1. Users whose
radial distance is between $\frac{d_1 + d_2}{2}$ and $\frac{d_2 + d_3}{2}$
belong to zone 2 and similarly, for the users whose radial distance is between
$\frac{d_2 + d_3}{2}$ and $d_3$ belong to zone 3. In approximating the
path-loss model with the $n$-power level IRSA model, we approximate that for
all the users in zone 1, the received power at the base-station is $P$ which is
equivalent to approximating the radial distance of all the users in zone 1 with
$r_i = d_1$ for all users $i$ in zone 1. Similarly we can approximate the radial
distance of all the users in zone 2 with $d_2$ and hence their received power
of all the users in zone 2 can be approximated to be $0.1P$ and the same hold
approximation is made for zone 3 users where their received power is
approximated to be $0.01P$. This approximation is shown in sub-Fig. (c).}
\end{figure}

\subsection{Closeness of approximation}
Due to the lack of any theoretical guarantees on the throughput of the vanilla
IRSA under the path loss model, we are unable to provide any theoretical
guarantees on how close the approximation of the throughput under the two
settings are. However we show through simulations that the throughput achieved
in the two settings i.e vanilla IRSA scheme under the path loss model and the
IRSA-nPC scheme under the ideal channel model are numerically very close.
We provide the simulation results for the case study described above in Fig.
\ref{fig:path_loss_sim}.

\begin{figure}[h]
\centering
\includegraphics[width = 0.6\linewidth]{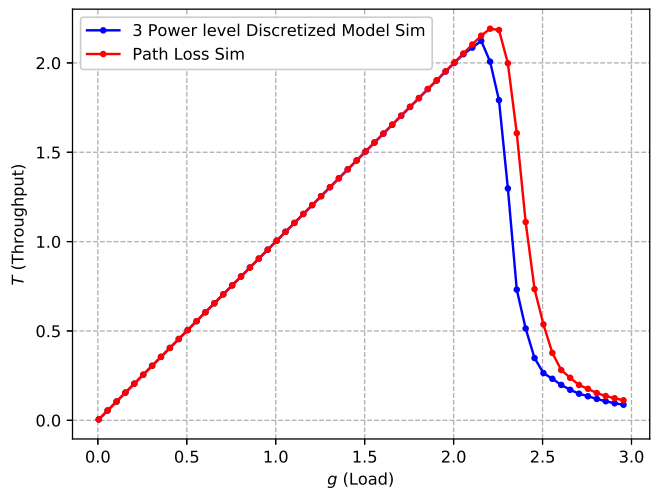}
\caption{For this simulation, we fix the number of
slots $M = 1000$. Each users
independently decided how many times to repeat her message according to the
distribution $\Lambda^{\text{liva}}(x) = 0.5x^2 + 0.28x^3 + 0.22x^8$. The radial
distance of
all the users is sampled independently from a uniform distribution over $
[0,d_{\max}]$. Other parameters for the simulation (which are mentioned in the
system model) are $\beta = 2, \alpha = 3, k = 5$.}
\label{fig:path_loss_sim}
\end{figure}
As is evident from the simulation results above, the throughput of
the IRSA-3PC scheme under the ideal channel model lower bounds the
throughput of the vanilla IRSA scheme under the path loss model. The intuition for this claim is as
follows. The set of powers received at the basestation under the path loss model
is $[P_{\min}, P]$ and the set of power received at the basestation with the
IRSA-nPC scheme under the ideal channel model is $\mathcal{P}$. We have that
$\mathcal{P} \subset [P_{\min}, P]$. Note that the condition $P_{i} \geq
k\beta P_{i-1}$ used in constructing the set $\mathcal{P}$ and for the analysis
of IRSA-nPC are only sufficient conditions and made for analytical tractibility.
It is possible for the packets to be decoded when the gap between the
subsequent power levels is not greater than $k \beta$, as we have assumed in the
IRSA-nPC scheme and this is true in the
case of the vanilla IRSA scheme under the path loss model and hence we observe
that more packets get decoded under the vanilla scheme with the path loss model
than in the IRSA-nPC scheme with the ideal channel model. 

\section{Conclusion}
 In prior work, for RAPs, using
redundancy in time (through repetition in IRSA), throughput was shown to
increase fundamentally over SA. We showed that by randomizing the transmit power for each packet,
another fundamental improvement can be obtained, and the throughout barrier of
unity can be breached, showing that there is no 'loss in throughput' because of lack of
coordination in the distributed system.
To this end, we introduced multi-level power control, where each user transmits its packet with power level chosen randomly according to a distribution. We set-up an optimization problem which enabled us to
optimize over the power-level distribution  and the RAP to
obtain the maximum throughput, where the complexity of the optimization problem is regulated by derivation of success probability evolution equations. We also provided upper bounds on the best possible throughputs possible, that were shown to be close to  
the achievable throughputs (solutions of the  optimization problems).


\bibliographystyle{ieeetr}
\bibliography{sample}

\vspace{-0.2in}
\begin{appendices}
\section{Proof of Lemma \ref{lm:p_irsa_2lvl}}\label{app:lemDE}
\begin{proof}
$q_1 =1$ because at the start of SIC decoding process, all edges connected to all user nodes are unknown.
At iteration $i$ of the SIC decoding process, consider a slot node with degree
$l$ as shown in Fig. \ref{subfig:slot-node}. WLOG, consider the edge $e_0$ and
consider $t$ other edges which are still unknown at iteration $i$. Therefore the
remaining $l - 1 - t$ edges are known (or resolved) by iteration $i$. Denote by
$w_{l,t}$ the probability that at a
particular iteration an edge (in Fig. \ref{subfig:slot-node}, edge $e_0$)
connected to a slot-node of degree $l$ with $t$ of the other edges being
unknown, becomes known at that iteration. If $t = 0$, then $l - 1$ have been
resolved and edge $e_0$ can be resolved by subtraction, hence $w_{l,0} = 1$. If
$t> k$, then $SIR < \beta$ and the edge $0$ can not be decoded since capture
effect can not be used and therefore $w_{l,t} = 0, \forall t > k$. If $t = 1$,
then we can resolve edge $0$ if and only if edge $1$ and edge $0$ have different
power levels i.e edge $0$ is $P_1$ and edge $1$ is $P_2$ and vice-versa. In
either of the case, due to the choice of $P_1,P_2$, we have that $SIR > \beta$
and we can leverage the capture effect to resolve both the edges. Therefore we
have that $w_{l,1} = 2\delta(1 - \delta)$. For the case of $t \in [2,k]$, the
only way edge $0$ can be resolved is if edge $0$ is $P_1$ and rest of the $t$
unknown edges are $P_2$ and we can resolve edge $e_0$ due to the capture effect
and we have that $w_{l,t} = \delta (1 - \delta)^t$ where $\delta$ is the
probability of edge $e_0$ being of power $P_1$ and $ (1 - \delta)^t$ corresponds
to the probability of $t$ unknown edges being of power $P_2$. Therefore, we have
that
\begin{equation}
w_{l,t} = \mathbbm{1}_{\{t = 0\}} + 2\delta(1 - \delta)\mathbbm{1}_{\{t = 1\}} +
\delta(1 - \delta)^t \mathbbm{1}_{\{t \in [2,\min\{k,l-1\}]\}}
\label{eq:wlt_2power_lvel}
\end{equation}
Let $1 - p$ denote the probability that edge $e_0$ is resolved after iteration
$i$. From the definition of $q_i$ and the fact that $t \in [0,l-1]$, it follows
that,
\begin{equation}
1 - p = \sum_{t = 0}^{l - 1} w_{l,t} {l - 1 \choose t} q_i^t(1 - q_i)^{l - t
- 1} = w_{l,0}(1 - q_i)^{l-1} + \sum_{t = 1}^{l-1} w_{l,t} {l-1 \choose t}q_i^t
(1-q_i)^{l - 1 - t},
\label{eq:1}
\end{equation}
where the term $w_{l,0}(1 - q_i)^{l-1}$ is due to interference cancellation and
$\sum_{t = 1}^{l-1} w_{l,t} {l-1 \choose t}q_i^t
(1-q_i)^{l - 1 - t}$ is due to the capture effect. Since we ignore noise in our model, $w_{l,0} = 1$. 
The term ${l-1 \choose t}$ refers to the number of
combinations of the $t$ other edges that are known.
Using the tree analysis argument presented in \cite{luby1998analysis}, by
averaging $p$ over the \emph{edge distribution} defined in Section 
\ref{subsec:graphical_model} and from the definition of $p_i$ it follows that
\begin{align}
p_i = \sum_{l = 1}^{N}\rho_l\left(1 - \sum_{t = 0}^{l - 1} w_{l,t} {l - 1 \choose t} q_i^t(1 - q_i)^{l - t
- 1} \right),
\label{eq:2}
\end{align}
where $\rho_l$ is the probability that an edge is connected to a slot node of
degree $l$. We simplify \eqref{eq:2} as follows $p_i $
\begin{align}\nonumber
         & \stackrel{(a)}{\approx} 1 - \sum_{l=1}^{N} \rho_{l} \left
        (\left(1 - q_i\right)^{l-1} + 
             \sum_{t=1}^{l-1} \delta (1 - \delta)^t 
            {l-1 \choose t} q_i^t (1- q_i)^{l-1-t} + \left(l - 1\right)\delta(1
        - \delta)q_i(1-q_i)^{l-2}\right), \\\nonumber
 & \stackrel{(b)}{=} 1 - \sum_{l=1}^{N} \rho_{l} \left(\left(1
            - q_i\right)^{l-1} + \delta\left(\left(1 - \delta q_i \right)^{l-1}
            - \left(1 -
            q_i\right)^{l-1}\right) + \left( l - 1\right)\delta(1 -
        \delta)q_i(1-q_i)^{l-2}\right), \\
       & \stackrel{(c)}{=} 1 - \left( 1 - \delta \right) e^{-gq_iR} -
        \delta e^{-gq_i\delta R} -
        \delta(1-\delta)gq_iRe^{-gq_iR},
    \end{align}
where (a) is obtained by substituting \eqref{eq:wlt_2power_lvel} in
\eqref{eq:1} and using the large $k$ approximation, $\min \{ k,l-1 \}$ =
$l-1$ (see Figure \ref{subfig:k_approximation} to see
the effect of this approximation), (b) is
obtained by using the binomial formula, and (c) is obtained by substituting the 
expression for the slot edge-perspective degree distribution, $\rho(x) = e^{-gR
(1-x)}$ as $N \to \infty$ because $M \to \infty$ for a constant load $g$.

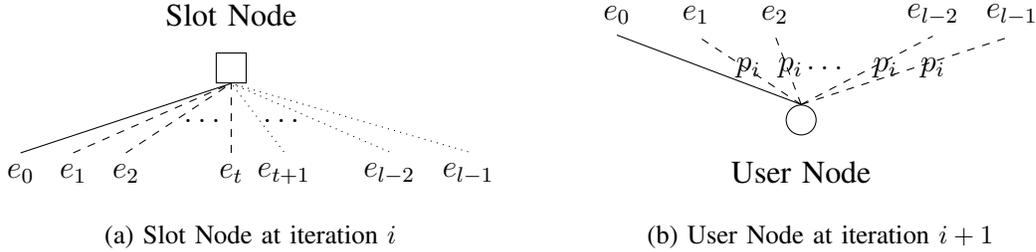
\begin{figure}[H]
    \begin{subfigure}{0.45 \textwidth}
    \centering
    \begin{tikzpicture}[scale = 0.7, square/.style={regular polygon,regular polygon sides=4}]
        \node at (0,2) [square,draw,minimum size = 0.25cm] (s1) {};
        \node at (0,3) {$\textrm{Slot Node}$};
        \node at (-4,0) (u1) {$e_0$};
        \node at (-3,0) (u2) {$e_1$};
        \node at (-2,0) (u3) {$e_2$};
        \node at (-0.5,1) (u8) {$\dots$};
        \node at (0,0) (u4) {$e_t$};
        \node at (1,0) (u5) {$e_{t+1}$};
        \node at (3,0) (u6) {$e_{l-2}$};
        \node at (4.5,0) (u7) {$e_{l-1}$};
        \node at (1,1) (u9) {$\dots$};
        \draw[] (u1.north) -- (s1.south);
        \draw[dashed] (u2.north) -- (s1.south);
        \draw[dashed] (u3.north) -- (s1.south);
        \draw[dashed] (u4.north) -- (s1.south);
        \draw[dotted] (u5.north) -- (s1.south);
        \draw[dotted] (u6.north) -- (s1.south);
        \draw[dotted] (u7.north) -- (s1.south);
    \end{tikzpicture}
    \caption{Slot Node at iteration $i$}
    \label{subfig:slot-node}
    \end{subfigure}
    \begin{subfigure}{0.45\textwidth}
	\centering
    \begin{tikzpicture}[scale = 0.7, square/.style={regular polygon,regular polygon sides=4}]
        \node at (2,0) [circle, draw, minimum size = 0.25 cm] (u2) {};
        \node at (2,-1) {$\textrm{User Node}$};
        \node at (-1.5,2) [] (s1) {$e_0$};
        \node at (0,2) [] (s2) {$e_1$};
        \node at (1.5,2) [] (s3) {$e_2$};
        \node at (2.5,1) [] (s6) {$\dots$};
        \node at (4.5,2) [] (s4) {$e_{l-2}$};
        \node at (6.0,2) [] (s5) {$e_{l-1}$};
        \node at (1,1) [] (t2) {$p_i$};
        \node at (1.8,1) [] (t3) {$p_i$};
        \node at (3.6,1) [] (t4) {$p_i$};
        \node at (4.5,1) [] (t5) {$p_i$};
        \draw[] (u2.north) -- (s1.south);
        \draw[dashed] (u2.north) -- (s2.south);
        \draw[dashed] (u2.north) -- (s3.south);
        \draw[dashed] (u2.north) -- (s4.south);
        \draw[dashed] (u2.north) -- (s5.south);
    \end{tikzpicture}
    \caption{User Node at iteration $i+1$}
    \label{subfig:user-node}
    \end{subfigure}
    \caption{Pictorial description of user and slot node with degree $l$}
\end{figure}
\vspace{-0.3in}
To derive \eqref{eq:q_iteration}, first consider a user node of degree $l$ at iteration
$i+1$ (see Fig. \ref{subfig:user-node}). 
Consider one particular edge $e_0$ out of the $l$ edges. Let each of the other $l-1$
edges remain
unknown on the slot-node side in the previous iteration $i$ independently with probability
$p_i$ each. Since the edge $e_0$ remains unknown at iteration $i+1$ only if all
the other $l-1$ edges are unknown this iteration, the probability that edge $e_0$ remains unknown at
iteration $i+1$ is $p_i^{l-1}$. 

Again using the tree analysis of \cite{luby1998analysis}, by averaging over the
edge distribution, we have that $q_{i + 1} = \sum_{l} \lambda_l p_i^{l - 1} =
\lambda(p_i)$, where recall that $\lambda_l$ denotes the probability that an
edge is connected to a user node of degree $l$ and $q_{i+1}$ is the probability
that an edge connected to a user node is unknown at iteration $i + 1$.
\end{proof}

\section{Proof of Lemma \ref{lem:monotone}}
\label{app:lem:monotone}
We consider two systems, where in System 1, the power levels are $\mathcal{P}
= \{P_1, P_2\}$ with $P_1 \geq k_1 P_2$, while in System 2, the power levels
are $\mathcal{Q} = \{Q_1, Q_2\}$ with $Q_1 \geq k_2 Q_2$, where we enforce that $k_1 \geq k_2$. 
In order to prove the result, we will
abstract out the randomness (slot location of packets etc.) in the two systems, i.e., we will assume two duplicate randomly chosen
instances (defined ${\mathsf R}$) of the problem (throughput maximization), and the only difference will be the power level of
the packets that are transmitted by the users. Then the instance for System 1 will be ${\mathsf R} \cup \{P_1, P_2\}$, while the same instance with System 2 will be ${\mathsf R} \cup \{Q_1, Q_2\}$. 

We next claim that for the given pair of instances (as demonstrated in Fig. 
\ref{fig:monotonicity}), if a packet is decoded in System 2, it must have been
decoded in System 1, using induction.
Consider a packet $p_0$ that gets decoded in the first
iteration of the SIC in System 2. Assume that packet $p_0$ was transmitted in slot $m$. In this case, there are two cases to consider. 
\begin{enumerate}
	\item \textbf{Case 1 :} The packet $p_0$ had no collision in slot $m$. Since ${\mathsf R}$ is the same, this implies that $p_0$ would not have any collisions in System 1 as well, and gets decoded with System 1.
	\item \textbf{Case 2 :} The packet $p_0$ had collision in slot $m$. But since $p_0$ got decoded in first iteration of SIC in 
	System 2,  packet $p_0$ is the only packet with power level $Q_1$ in slot $m$, and there are at most $k_2$ other packets transmitted in the  slot $m$ with power level $Q_2$. Thus, since ${\mathsf R}$ is the same, in System $1$, we have $p_0$ as the only packet with power level of $P_1$ in slot $m$, and there are at most $k_2$ other packets transmitted in the  slot $m$ 
	with power level $P_2$. Since $k_1 \geq k_2$ by choice i.e., System 1 has more tolerance for interference, we get that $p_0$ will get decoded in the first
iteration of the SIC in System 1 as well.
\end{enumerate}

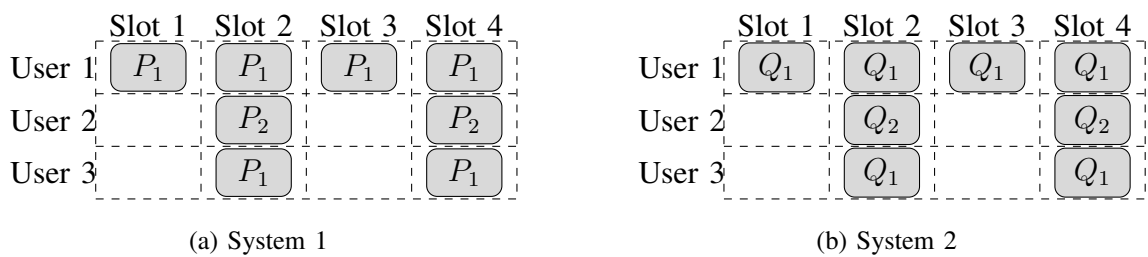
\begin{figure}[H]
	\begin{subfigure}[b]{0.5\textwidth}
	\centering
	\begin{tikzpicture}[scale = 0.7]
		\foreach \x in {0,...,-3}
		\draw[dashed] (0,\x) -- (8,\x);
		\foreach \x in {0,2,4,6,8}
		\draw[dashed] (\x,0) -- (\x,-3);
		\foreach \x in {1,2,3,4}
		\node at (2*\x-1,0.3) {Slot \x};
		\foreach \x in {1,2,3}
		\node at (-0.8, -\x + 0.5) {User \x};
		\foreach \x in {1,2,3,4}
		\node [draw, fill= gray!30, shape=rectangle, minimum width=1cm, minimum
        height=0.5cm, anchor=center, rounded corners] at (2*\x-1,-0.5) {$P_1$};
		\foreach \x in {2,4}
		\node [draw, fill= gray!30, shape=rectangle, minimum width=1cm, minimum
        height=0.5cm, anchor=center, rounded corners] at (2*\x-1,-1.5) {$P_2$};
		\foreach \x in {2,4}
		\node [draw, fill= gray!30, shape=rectangle, minimum width=1cm, minimum
        height=0.5cm, anchor=center, rounded corners] at (2*\x-1,-2.5) 
        {$P_1$};
	\end{tikzpicture}
	\caption{System 1}
	\end{subfigure}
	\begin{subfigure}[b]{0.5\textwidth}
	\centering
	\begin{tikzpicture}[scale = 0.7]
		\foreach \x in {0,...,-3}
		\draw[dashed] (0,\x) -- (8,\x);
		\foreach \x in {0,2,4,6,8}
		\draw[dashed] (\x,0) -- (\x,-3);
		\foreach \x in {1,2,3,4}
		\node at (2*\x-1,0.3) {Slot \x};
		\foreach \x in {1,2,3}
		\node at (-0.8, -\x + 0.5) {User \x};
		\foreach \x in {1,2,3,4}
		\node [draw, fill= gray!30, shape=rectangle, minimum width=1cm, minimum
        height=0.5cm, anchor=center, rounded corners] at (2*\x-1,-0.5) 
        {$Q_1$};
		\foreach \x in {2,4}
		\node [draw, fill= gray!30, shape=rectangle, minimum width=1cm, minimum
        height=0.5cm, anchor=center, rounded corners] at (2*\x-1,-1.5) 
        {$Q_2$};
		\foreach \x in {2,4}
		\node [draw, fill= gray!30, shape=rectangle, minimum width=1cm, minimum
        height=0.5cm, anchor=center, rounded corners] at (2*\x-1,-2.5) 
        {$Q_1$};
	\end{tikzpicture}
	\caption{System 2}
	\end{subfigure}
	\caption{Note that this instance of System 1 and System 2 only differ in
	the power levels.}
	\label{fig:monotonicity}
\end{figure}
\vspace{-0.3in}

\noindent 
\textbf{Induction Hypothesis : } Assume that if a packet $p$ is decoded by
iteration $t-1$ in System 2, then it is also decoded by iteration $t-1$ in System 1 as
well. (Induction Step :)  Now, consider a packet $p_t$ that got decoded at
iteration $t$ in System 2. Then, in System 1 we have 2 cases corresponding
to packet $p_t$:
    \begin{enumerate}
        \item \textbf{Case 1 : }$p_t$ got decoded before iteration $t$ in
        System 1.
        \item \textbf{Case 2 : }$p_t$ is not yet decoded by iteration $t$ in System
        1. Then we have 2 subcases:
            
            \begin{enumerate}
                \item \textbf{Case 2(a) : } If $p_t$ had no interferers at
                iteration $t$ in System 2: then by the induction hypothesis it has no interferers at
                iteration $t$ in System 1 as well. Therefore, it gets decoded in System 1 as well. 
                \item \textbf{Case 2(b) : } If $p_t$ had power $Q_1$ in System
                2 and had at most $k_2$ interferers with lower power $Q_2$:
                then, in System 1, the packet $p_t$ has power $P_1$, and again
                by induction hypothesis, it has at most $k_2$ interferers all
                with power $P_2$. Therefore, $p_t$ gets decoded at iteration
                $t$ in System 1 as well.
            \end{enumerate}
    \end{enumerate}

    This completes the proof by induction that under a fixed instantiation, if a packet $p$ gets decoded in System 2 by iteration $t$, then it is
    decoded by iteraton $t$ in System 1 as well.
    This means that, by the end of the SIC decoding process on a fixed instantiation, the number of
    packets decoded in System 1 is atleast as much as the number of
    packets decoded in System 2.
    Therefore, averaging over all the instantiations, we get 
        $T(g, M, \P^1, \delta, \Lambda ) \geq T(g, M, \P^2, \delta,
        \Lambda)$.

\section{Proofs for Section \ref{sec:SA}}

\label{appendix:proof_SA}

\begin{proof}[Proof for Lemma \ref{lm:k_to_inf_approx}]
The slots chosen by each of the users are independent and are uniformly
at random. Let $\textrm{Pr}_j(k)$ denote the probability that there are $k$
packet transmissions in slot $j$:
\begin{align*}
\textnormal{Pr}_j(k) &\stackrel{(a)}= {N \choose k} \left( \frac{1}{M} \right)^
{k} \left(1 - \frac{1}{M} \right)^{gM-k} \\
	&\stackrel{(b)}{\approx} \frac{g^{k}e^{-g}}{k!} \\
    &\stackrel{(c)}{\leq} \frac{g^m}{k(k-1)\dots(k-m+1)} \frac{g^{k-m}}{
    (k-m)!}e^{-g} \\
    &\stackrel{(d)}\leq \frac{K}{k^m} \\
    &\stackrel{(e)}= \mathcal{O}\left( \frac{1}{k^m} \right)
\end{align*}
where (a) follows from the fact that there are ${N \choose k}$ ways of having
$k$ packets transmitted and the probability of each such event is $\left( 
\frac{1}{M}\right)^{k} \left(1 - \frac{1}{M} \right)^{gM-k}$ (b) follows from the Stirlings'
approximation for large $N$, (c) follows for any finite $m$, (d) is
true for a suitably large enough constant $K$ and follows from the fact that $\frac{g^
{k-m}}{(k-m)!} \leq e^{g}$, (e) follows from the definition of $\mathcal{O}(.)$. 
Since this is true for any finite $m$, we get that $\text{Pr}_j(k) = 
\mathcal{O}\left(\frac{1}{ \textrm{poly}(k)}\right)$.
\end{proof}

\begin{proof}[Proof of Lemma \ref{lm:count_capture}] 
From the capture effect, it follows that a packet of higher power $P_1$ can be
decoded if there are at most $k$ interfering packets with lower power $P_2$. Let $A$ denote
the event that packet of higher power $P_1$ gets decoded when there are atmost
$k$ interferring packets. Let $B_i$ denote the event that one of the transmitted
packets is with power $P_1$ and there are exactly $i$ transmitted packets with power
$P_2$. It follows that $A = \cup_{i=0}^k B_i$ and $B_i \cap B_j = \phi, \forall
i,j$. It follows that
\begin{align}
\nonumber,
\mathbb{P}(A) &= \mathbb{P}\left(\cup_{i=0}^k B_i \right), \\
\nonumber
&\stackrel{(a)}= \sum_{i=0}^k \mathbb{P}(B_i), \\
\nonumber
&\stackrel{(b)}= \sum_{i=0}^k {N \choose i+1}\left(\frac{1}{M}\right)^{i+1}\left
(1 - \frac{1}{M} \right)^{N - i -1}{i+1 \choose 1}\delta(1 - \delta)^{i}, \\
\nonumber
&\stackrel{(c)}\approx \sum_{i=0}^\infty {N \choose i+1}\left(\frac{1}
{M}\right)^{i+1}\left(1 - \frac{1}{M} \right)^{N - i -1}{i+1 \choose 1}\delta(1
- \delta)^{i}, \\
\nonumber
&\stackrel{(d)}\approx \sum_{i=0}^\infty \frac{g^i}{i!}e^{-g}\delta(1-\delta)^i,
\\
\nonumber
&\stackrel{(e)}= g\delta e^{-g\delta},
\end{align}
where (a) follows from fact that $B_i \cap B_j = \phi$, (b) follows from the
fact that the definition of the event $B_i$, (c) follows from
Lemma \ref{lm:k_to_inf_approx} since $k$ is assumed to be sufficiently large, (d) follows
from the Stirlings' approximation for the term ${N \choose i+1}$ and $g = N/M$, 
(e) follows from the Taylor Series expansion of $e^x$.
\end{proof}

\begin{proof} [Proof of Theorem \ref{theorem:dpc_sa_thruput}]
Let $\bar{N}_1$ denote the average number with power level packet with power $P_1$s decoded
per slot and $\bar{N}_2$ denote the average number with power level packet with power $P_2$s
decoded. Then $T = \bar{N}_1 + \bar{N}_2$. From Lemma \ref{lm:count_capture} it
follows that $\bar{N}_1 = g\delta e^{-g \delta}$. For calculating the value of
$\bar{N}_2$, there are two cases where a lower power level packet with power $P_2$ is
decoded : (a) the lower power level packet is the only packet in a given slot, 
(b) there are two packets in a given slot - one of higher power level $P_1$ and
one of lower power level $P_2$. In case (b), the $P_1$ power level packet can be
decoded using the capture effect and using SIC, the higher power level packet
can be ``subtracted" to decode the lower power level packet. Let $\bar{N}_
{2a},\bar{N}_{2b} $ denote the average number of packets decoded in case (a) and
case (b) respectively. 
\begin{align*}
\bar{N}_2 &= \bar{N}_{2a} + \bar{N}_{2b} \\
&\stackrel{(a)}= {N \choose 1}\left( \frac{1}{M}\right)\left(1 - \frac{1}{M}
\right)^{gM - 1}\left( 1 - \delta\right) + 2{N \choose 2}\left( \frac{1}
{M}\right)^2\left(1 - \frac{1}{M}\right)^{gM - 2}\delta\left( 1 - \delta\right),
\\
&\stackrel{(b)}= g(1- \delta)e^{-g}\left(1 + g\delta\right),
\end{align*}
where (a) follows from the fact that $\bar{N}_{2a}$ is the probability that only
one low power packet transmission happens in a given slot and $\bar{N}_{2b}$ is
the probability that exactly two packets get transimitted and they are of
different power levels, (b) follows from taking $N \to \infty, M \to \infty$.
Combining the results from Lemma \ref{lm:count_capture} and calculation of
$\bar{N}_2$, $T_{\text{SA-DPC}} = g\delta e^{-g\delta} + (1 + g\delta)g
(1-\delta)e^{-g}$.
\end{proof}

\begin{proof} [Proof of Theorem \ref{theorem:npc_sa_thruput}]
Let $\bar{N}_i, \forall i \in [n-1]$ denote the average number of packet of
power level $P_i$ decoded. Then $T_{\text{nPC-SA}} = \sum_{i=1}^n \bar{N}_i$.
$\bar{N}_1 = g\delta_1 e^{-g \delta_1}$ follows from Lemma 
\ref{lm:count_capture}, $\bar{N}_2 = (1 + g\delta_1)g\delta_2 e^{-g(\delta_1
+ \delta_2)}$  follows from taking two cases - (a) there are at most $k$ lower
power level packets, (b) there are one packet with power level $P_1$ and at most
$k$ lower power level  packets. The calculation of $\bar{N}_2$ is along the same
lines as in the proof of Theorem \ref{theorem:dpc_sa_thruput}. For computing
$\bar{N}_3$, there are 4 cases: 
\begin{enumerate}[label = (\alph*)]
\item $0$ $P_1$ power level packets, $0$ $P_2$ power level packets and atmost
$k$ packets with power level $P_i, i > 3$.
\item $0$ $P_1$ power level packets, $1$ $P_2$ power level packets and atmost
$k$ packets with power level $P_i, i > 3$.
\item $1$ $P_1$ power level packets, $0$ $P_2$ power level packets and atmost
$k$ packets with power level $P_i, i > 3$.
\item $1$ $P_1$ power level packets, $1$ $P_2$ power level packets and atmost
$k$ packets with power level $P_i, i > 3$. 
\end{enumerate}
Summing across the 4 cases gives us that $\bar{N}_3 = \left( 1 + g \delta_1
\right)\left( 1 + g\delta_2 \right)g\delta_3 e^{-g(\delta_1 + \delta_2 +
\delta_3)}$. Consider the case of $\bar{N}_i$, there are $2^{i-1}$ cases to
consider which are similar to the cases considered for computing $\bar{N}_3$. In
those $2^{i-1}$ we consider all combinations of $0$ and $1$ packets with power
level $P_j, j < i$ and atmost $k$ packets with power level $P_j, j > i$. This
count is given as $\bar{N}_i = (1+g\delta_1)(1+g\delta_2)\dots(1+g\delta_
{i-1})g\delta_i e^{-g(\delta_1 + \delta_2 + \dots + \delta_i)} = \left(\prod_
{j=1}^{i-1}(1 + g\delta_j)e^{-g\delta_j}\right)g\delta_i e^{-g \delta_i}$.
Summing all these $\bar{N}_i$ gives us the required result.
\end{proof}

\section{}
\label{appendix:irsa-npc}
\subsection{IRSA-PC with Three Power Levels: IRSA-3PC}
\begin{lemma}
    Consider a system $\left\{ M, \P, \Delta, \beta \right\}$ with three power levels $\mathcal{P} = \{P_1,
P_2, P_3\}$ with the corresponding power choice distribution $\Delta = \{\delta_1,
\delta_2, \delta_3\}$. $M$ is the number of slos in one frame, and $\beta$ is the capture threshold. The power levels are selected such that $P_1 \geq k_1
\beta P_2, P_2 \geq k_2\beta P_3$ and  $k_1,k_2$ are sufficiently large positive integers.
Let the IRSA-PC scheme be given by the repetition distribution $\Lambda(x)$(which uniquely specifies
    the edge-perspective distributions $\lambda(x) = \sum \lambda_l x^{l-1}$ and $\rho(x) = \sum
\rho_l x^{l-1}$). 
Let
$q_i$ and $p_i$ denote the probability that an edge connected to a user node and
a slot node remain unknown respectively at iteration $i$ of SIC. Then as $M \to
\infty$, we have that 
\begin{equation}
\begin{split}
q_1 &= 1, \\
\forall i &\geq 1: \\
p_i &= 1 - \rho_1 - \rho_2\left((1 - q_i) + \sum_{j = 1}^3 \delta_j \left(1 -
\delta_j \right)q_i \right) - \sum_{l=3}^{N} \rho_l (1 - q_i)^{l-1}  \\ 
&\qquad - \sum_{l=3}^{N} \rho_l \sum_{t = 1}^{l-1} \delta_1(1 - \delta_1)^t {l-1 \choose t}q_i^t(1-q_i)^{l-t-1} \\ 
&\qquad - \sum_{l=3}^{N} \rho_l \sum_{t = 1}^{l-1} \delta_2\left(\delta_3^t + {t \choose 1}\delta_1
\delta_3^{t-1} \right){l-1 \choose t}q_i^t(1-q_i)^{l-t-1}\\ 
&\qquad - \sum_{l=3}^{N}\rho_{l}\left(\delta_3 \left( 1 - \delta_3 \right){l-1
\choose 1} \left( 1- q_i\right)^{l-2}q_i + 2\delta_1 \delta_2 \delta_3 {l-1
\choose 2}\left(1 - q_i \right)^{l-3}q_i^2 \right) , \\
q_{i+1} &= \lambda(p_i).
\label{eq:3pow_lvl_2} 
\end{split}
\end{equation}
\label{lm:p_irsa_3lvl}
\end{lemma}
The proof of Lemma \ref{lm:p_irsa_3lvl} follows on the same lines as Lemma 
\ref{lm:p_irsa_2lvl}. A brief sketch is provided as follows: 

\begin{itemize}
    \item Slote node with degree 1($\rho_1$ term): At this node, there are no collisions. Therefore it is decoded
        with probability 1.
    \item Slot node with degree 2($\rho_2$ term): At this node, the selected edge will be decoded in
        the current iteration as long as the other edge has a packet of a different power-level.
    \item Slot node with degree greater than 2:
        \begin{itemize}
            \item Term with $\delta_1 (1-\delta_1)^t$: this term corresponds to the case where the
                selected edge has a power of $P_1$, and the rest of the unkown edges have a lower power. In
                this case, the selected edge can be decoded in this iteration.
            \item Term with $\delta_2\left(\delta_3^t + {t \choose 1}\delta_1
                \delta_3^{t-1} \right)$: this corresponds to the case where the selected edge has a
                power of $P_2$. In this case, the selected edge is decoded in the current iteration
                if all the other unknown edges are of the lower power $P_3$, or if one of them is of
                power $P_1$ and the rest of them are with power $P_3$.
            \item Term with $\delta_3(1-\delta_3)$: This is the case where the selected edge has the
                lowest power $P_3$, and there is only one other unknown edge connected to the
                node, and it has a higher power of $P_1$ or $P_2$.
            \item Term with $\delta_1 \delta_2 \delta_3$: This is the case where the selected edge
                has the lowest power $P_3$, and the there are two other unknown edges connected to
                the node, with one with power $P_1$ and the other with power $P_2$.
        \end{itemize}
\end{itemize}

Note that the expression for $p_i$ in Lemma \ref{lm:p_irsa_3lvl} is rather
cumbersome and does not simplify into a closed form expression as in the case of
Lemma \ref{lm:p_irsa_2lvl}. As the value of $n$ increases, calculating a closed
form expression for $p_i$ become more and more complicated. Hence in the next
section, we discuss a recursive algorithm for calculating the coefficients $w_
{l,t}$ required in the expression for $p_i$.

\subsection{General $n$ Power Level with IRSA (IRSA-nPC)}
\label{subsec:general-n-power-level-irsa}
In this section we will set up DE equations for IRSA in systems with $n$ power
levels, which we will call IRSA-nPC. Let the system $\left\{ M, \P, \Delta, \beta, \Lambda \right\}$ be given by
the set with power levels $\P = \left\{ P_1, P_2, \dots, P_n \right\} $ with corresponding power
choice distribution $\Delta = \left\{ \delta_1, \delta_2, \dots, \delta_n \right\}$. 
As in the
previous sections, we assume that there is a sufficient multiplicative gap
between successive power levels: $P_{i} > k \beta P_{i-1}$, for some sufficiently large integer $k$. We need to
basically compute $w_{l,t}$, which is the probability that a a packet connected to a degree
$l$ slot with $t$ unresolved packets gets resolved at this particular time. We will express it as,

\begin{equation}
w_{l,t} = \sum_{i=1}^{n} \delta_{i} \sum_{j=0}^{t} {t \choose j} j! Pr_{i,j}^{\text{higher}} Pr_{i,t-j}^{\text{lower}}.
\label{eq:w_l_t_irsa_npc}
\end{equation}

This packet could be of any power, which is represented by $\sum \delta_{i}(\cdot)$ in
the above formula. Then it could have anything from 0 to $t$ higher-power
interferers( represented by $j$ in the above formula) and $t-j$ lower power
interferers respectively. The term $Pr_{i,j}^{\text{higher}}$ represents the
probability that a node with Power $i$ has $j$ higher power interferers, such
that it still gets decoded at this particular step. Similarly,
$Pr_{i,t-j}^{\text{lower}}$ represents the probability that a packet with power
$P_i$ has $t-j$ lower power interferers such that the chosen packet still gets
decoded. The ${t \choose j}$ term represents the choice of $j$ locations for
higher power interferers, and $j!$ represents the permutations in locations of
the higher power interferers possible. We will now show how to compute
$Pr_{i,t-j}^{\text{lower}}$ and $Pr_{i,j}^{\text{higher}}$.

\paragraph{Computing $Pr_{i,t-j}^{\text{lower}}$}
Computation of $Pr_{i,t-j}^{\text{lower}}$ is simple because of our
multiplicative gap assumption: $P_{i} > k\beta P_{i-1} $. With this assumption,
the chosen packet can get "captured" for any number of low power interferers
$t-j$. This term is therefore simply calculated as:

\begin{equation}
Pr_{i,t-j}^{\text{lower}} = (\delta_{i+1}+\ldots+\delta_{n})^{t-j}.
\end{equation}

\paragraph{Computing $Pr_{i,j}^{\text{higher}}$}

Computing this quantity is a bit more involved. Before actually giving the method of computation, let's first look at how the chosen packet gets decoded in the presence of $j$ higher power interferers. If a packet with power $P_{i}$ has $j$ higher power interferers, then all these $j$ higher power interferers need to be first captured before the capture of the chosen packet is possible. All the $j$ higher power interferers can be captured in a single iteration step if and only if they all have distinct power levels. This can be seen easily by a contradiction example: if there is there are two interferers with the same power level, then neither of them have $SIR \geq \beta$ (since $\beta > 1$). Therefore, neither of them can be captured. We use this fact in computing $Pr_{i,j}^{\text{higher}}$. 

Firstly, it is easy to see that if $i \leq j$, it is not possible for the packet to have $j$ unique higher power interferers. Thus, $Pr_{i,j}^{\text{higher}} = 0, \qquad \text{if } i \leq j $. We will do the computation using an iterative algorithm. It is easy to compute the following two quantities, which denote probabilities of $0$ and $1$ higher power interferers respectively:

\begin{equation}
\begin{split}
Pr_{i,0}^{\text{higher}} &= 1 ,\\
Pr_{i,1}^{\text{higher}} &= \delta_{1} + \ldots \delta_{i-1}.
\end{split}
\label{eq:de_initial_condtn}
\end{equation}

Now, $Pr^{\text{higher}}_{i,j}$ for a general $j$ can be computed using the quantities $Pr^{\text{higher}}_{i-1,j-1}$ and $Pr^{\text{higher}}_{i-1,j}$. Notice that higher power interferers for a packet with power $P_{i}$ are all the higher power interferers for a packet with power $P_{i-1}$ and the packet with power $P_{i-1}$ itself. Therefore, $Pr^{\text{higher}}_{i,j}$ can be computed as:

\begin{equation}
Pr^{\text{higher}}_{i,j} = Pr^{\text{higher}}_{i-1,j} + \delta_{i-1} Pr^{\text{higher}}_{i-1,j-1}.
\end{equation}

With the initial conditions above and this iterative process, we can compute
$Pr^{\text{higher}}_{i,j}$ for all $i,j$. We have thus laid down a procedure to compute $w_{l,t}$ for a general $n$-level power control model for CSA. \\

\begin{algorithm}[H]
\label{alg:compute_de_coefficients}
\SetAlgoLined
\KwResult{Returns the $w_{l,t}$}
$w_{l,t} \gets 0$ \;
 \For{i $\gets$ 1 to $n$}{
    $Pr_{i}^{\text{cap}} \gets 0$ \;
    \For{j $\gets$ 1 to $t$}{
       $Pr_{i,t-j}^{\text{lower}} \gets (\delta_{i+1}+\ldots + \delta_{n})^
       {t-j}$ \;
       \uIf{$j=0$}{
        $Pr_{i,0}^{\text{higher}} \gets 1$\;
       }
       \uElseIf{$j=1$}{
        $Pr_{i,1}^{\text{higher}} \gets \delta_{1} + \ldots + \delta_{i-1}$\;
       }
       \Else{
        $Pr^{\text{higher}}_{i,j} \gets Pr^{\text{higher}}_{i-1,j} + \delta_
        {i-1} Pr^{\text{higher}}_{i-1,j-1}$ \;
       }
       $\text{$Pr_{i}^{\text{cap}}$} \gets \text{$Pr_{i}^{\text{cap}}$} + {t
       \choose j} j! Pr_{i,j}^{\text{higher}} Pr_{i,j}^{\text{lower}}$ \;
    }   
    $w_{l,t} \gets w_{l,t} + \delta_{i}\times\text{$Pr_{i}^{\text{cap}}$}$
}
 \caption{Algorithm to compute $w_{l,t}$ for $n$-level power control}
\end{algorithm}

\end{appendices}
\end{document}